\documentclass[12pt,journal,draftcls,onecolumn]{IEEEtran}
\usepackage[T1]{fontenc}

\ifCLASSINFOpdf

\else

\fi
\usepackage{amsmath}
\usepackage{ifthen}
\usepackage{tikz}
\usepackage{cite}
\usepackage{titlesec}
\renewcommand\thesubsubsection{\arabic{subsubsection}.}
\titleformat{\subsubsection}
{\normalfont\normalsize\itshape}{\thesubsubsection}{0.4em}{}
\titlespacing*{\subsubsection}
{0pt}{1.25ex plus 1ex minus .2ex}{1.25ex plus .2ex}

\usepackage{amsthm,amsfonts,amssymb}
\usepackage{graphicx}
\usepackage{setspace} 
\usepackage[left=1in,right=1in,top=1in,bottom=1in,margin=1in]{geometry}
\usepackage[most]{tcolorbox}
\newtcolorbox{mybox}[2][]{%
	attach boxed title to top center
	colbacktitle = white,
	title        = #2,
	enhanced,
}
\usepackage{subfig}
\usepackage{blkarray}
\usepackage{dsfont}
\usepackage[mathscr]{euscript}
\usepackage{enumitem}
\usepackage{algorithm}
\usepackage[noend]{algpseudocode}
\usepackage{blkarray}
\usepackage{stfloats}%
\newtheorem{theorem}{Theorem}[section]
\newtheorem{corollary}{Corollary}[section]
\newtheorem{proposition}{Proposition}[section]
\newtheorem{lemma}{Lemma}[section]
\theoremstyle{remark}
\newtheorem*{remark}{Remark}
\theoremstyle{definition}
\newtheorem{definition}{Definition}[section]

\usepackage{url}
\usepackage{arydshln}
\usepackage{mathtools}
\usepackage{dsfont}
\usepackage{graphicx}

\definecolor{brickred}{cmyk}{0,0.89,0.94,0.28}
\definecolor{goldenrod}{cmyk}{0,0.10,0.84,0}
\definecolor{purple}{cmyk}{0.45,0.86,0,0}
\definecolor{rawsienna}{cmyk}{0,0.72,1,0.45}
\definecolor{olivegreen}{cmyk}{0.64,0,0.95,0.40}
\definecolor{peach}{cmyk}{0,0.5,0.7,0}
\definecolor{darkolive}{rgb}{0.,0.4,0.}
\colorlet{grey}{gray!40}

\makeatletter
\newcommand{\ostar}{\mathbin{\mathpalette\make@circled\star}}
\newcommand{\make@circled}[2]{%
	\ooalign{$\m@th#1\smallbigcirc{#1}$\cr\hidewidth$\m@th#1#2$\hidewidth\cr}%
}
\newcommand{\smallbigcirc}[1]{%
	\vcenter{\hbox{\scalebox{0.77778}{$\m@th#1\bigcirc$}}}%
}
\makeatother

\usepackage[cmintegrals]{newtxmath}

\begin{document}

\title{Estimating the Sizes of Binary Error-Correcting Constrained Codes}
%
%
%

\author{V.~Arvind~Rameshwar,~\IEEEmembership{Student Member,~IEEE,}
        and~Navin~Kashyap,~\IEEEmembership{Senior~Member,~IEEE}
\thanks{This work was supported in part by a Qualcomm Innovation Fellowship India 2022. The work of V.~A.~Rameshwar was supported by a Prime Minister's Research Fellowship, from the Ministry of Education, Govt. of India.}
\thanks{The authors are with the Department of Electrical Communication Engineering, Indian Institute of Science, Bengaluru 560012, India (e-mail: vrameshwar@iisc.ac.in;~nkashyap@iisc.ac.in).}
}


\maketitle

\begin{abstract}

In this paper, we study binary constrained codes that are resilient to bit-flip errors and erasures. 
In our first approach, we compute the sizes of constrained subcodes of linear codes. Since  there exist well-known linear codes that achieve vanishing probabilities of error over the binary symmetric channel (which causes bit-flip errors) and the binary erasure channel, constrained subcodes of such linear codes are also resilient to random bit-flip errors and erasures. We employ a simple identity from the Fourier analysis of Boolean functions, which transforms the problem of counting constrained codewords of linear codes to a question about the structure of the dual code. We illustrate the utility of our method in providing explicit values or efficient algorithms for our counting problem, by showing that the Fourier transform of the indicator function of the constraint is computable, for different constraints. Our second approach is to obtain good upper bounds, using an extension of Delsarte's linear program (LP), on the largest sizes of constrained codes that can correct a fixed number of combinatorial errors or erasures. We observe that the numerical values of our LP-based upper bounds beat the generalized sphere packing bounds of Fazeli, Vardy, and Yaakobi (2015).

\end{abstract}


\IEEEpeerreviewmaketitle

\section{Introduction}
{C}{onstrained} coding is a method that is employed in several domains such as magneto-optical recording (see, for example, \cite{Roth} or \cite{Immink}), DNA data storage \cite{dna1,dna2}, and energy harvesting communication \cite{subblock1,infoenergy}, which  allows the encoding of arbitrary user data sequences into only those sequences that respect a certain constraint. Our interest in this paper is in constrained codes that are also resilient to symmetric errors and erasures.

In our first approach, we consider the setting of the transmission of constrained codes over a noisy channel that induces errors stochastically. 
It is in this context that we consider the transmission of constrained subcodes of binary linear codes: if the channel introducing errors or erasures is memoryless and symmetric, there are well-known binary linear codes that achieve the capacity or whose rates are very close to the capacity of the channel (see \cite{polar, Reeves,ru1}). In particular, this means that constrained subcodes of such linear codes also enjoy vanishing error probabilities over such binary-input memoryless symmetric (BMS) channels (which include the binary symmetric channel (BSC) that introduces bit-flip errors and the binary erasure channel (BEC)), in the limit as the blocklength goes to infinity. This observation can be useful for the construction of error-correcting constrained coding schemes of good rates over input-constrained BMS channels, without feedback \cite{arnk22titarxiv}. We mention here that the problem of {explicitly} constructing capacity-achieving codes over input-constrained memoryless channels, which form a special class of finite-state channels (FSCs), is still wide open. {While it is known from the results of \cite{sasoglutal1} and \cite{litan} (see also \cite{shuvaltal}) that polar codes achieve the capacity of such input-constrained DMCs, the results therein do not explicitly identify  the polar code bit-channels to send information bits over, and hence do not give rise to an explicit construction.}


As part of our first approach, we are interested in the problem of determining the sizes of constrained subcodes of linear codes. Our approach to this counting problem makes use of a simple identity from the Fourier analysis of Boolean functions, namely, Plancherel's Theorem, which transforms our counting problem to one in the space of the dual code. An immediate advantage of this approach is that the dimension of the vector space over which we count, which is the minimum of the dimensions of the linear code and its dual, is always bounded above by half the blocklength of the code. Our study reveals the somewhat surprising fact that for many constraints, the Fourier transform of the indicator function of the constraint is computable, either analytically, or via efficient recursive procedures---an observation that can be of independent theoretical interest. We show, using specific examples of constraints, that our approach can yield not just the values of the sizes of constrained subcodes of specific linear codes, but also interesting insights into the construction of linear codes with a prescribed number of constrained codewords. 


Next, in our second approach, we consider the situation when the constrained codewords we seek to transmit or store are subject to adversarial (or combinatorial) errors or erasures. More precisely, we are interested in the scenario where there is an upper bound on the \emph{number} of errors or erasures that the channel can introduce (with potential adversarial knowledge of the codeword as well), and we would like to recover our constrained codeword with zero error. It is well-known (see, for example, \cite{roth_coding_theory}) that the minimum Hamming distance (or simply, minimum distance) of the constrained code determines the number of such errors or erasures that it can tolerate. Hence, we seek to obtain good bounds on the sizes of constrained codes with a prescribed minimum distance. 

There is extensive literature on the construction of and bounds for constrained codes with a certain minimum distance, and we refer the reader to Chapter 9 of \cite{Roth} for references. In particular, \cite{Ferreira} provided a simple lower bound on the sizes of runlength-limited (RLL) constrained codes with a given minimum Hamming distance, by a coset-averaging argument for linear codes. These bounds were then improved upon by Kolesnik and Krachkovsky \cite{kolesnik} and Marcus and Roth \cite{marcusroth}, via the solutions to certain constrained optimization problems. Less was known in the case of upper bounds on constrained codes with a given minimum distance, until the works of Cullina and Kiyavash \cite{cullina} (see also \cite{kulkarni}) and Fazeli, Vardy, and Yaakobi \cite{fazeli}, which provided a generalization of the well-known sphere packing bound for codes, to the setting of constrained codes\footnote{While these papers were focussed on obtaining bounds on the sizes of codes for combinatorial error models, their techniques can be easily applied to determining upper bounds on the sizes of constrained codes with a given minimum distance as well.}. The approach in \cite{cullina} and \cite{fazeli} was based on finding the size of the largest matching, or equivalently, the size of the smallest transversal, in a suitably defined hypergraph. 

In this paper, we provide a different approach to deriving good upper bounds on the sizes of constrained codes with a given minimum Hamming distance, by modifying Delsarte's well-known linear program (LP) \cite{delsarte} to the setting of constrained systems. While on a first pass, we propose an LP whose number of variables is exponential in the blocklength of the code, we show that for certain constraints, it is possible to ``symmetrize'' this LP to derive an equivalent LP with much smaller numbers of variables and constraints, which are sometimes only polynomial functions of the blocklength. We use our LPs to numerically calculate upper bounds on the sizes of the largest constrained codes with a prescribed minimum distance, for different constraints, and show that the values we obtain by our approach beat those obtained via the generalized sphere packing bounds. 

The remainder of the paper is organized as follows: Section \ref{sec:preliminaries} puts down notation and refreshes some background on binary codes and elementary Fourier analysis on the hypercube. Section \ref{sec:main} introduces the main theorem for counting constrained codewords in linear codes and our LPs for upper bounding the sizes of constrained codes with a given minimum distance. Section \ref{sec:eg} then applies the theorems and LPs discussed in Section \ref{sec:main} to different constraints. The paper is  concluded in Section \ref{sec:conclusion} with some directions for future work.
\section{Preliminaries}
\label{sec:preliminaries}
\subsection{Notation}
\label{sec:notation}
Given integers $a, b$, we use the notation $[a:b]$ to denote the set $\{a,a+1,\ldots,b\}$, for $a\leq b$. We use $[n]$ as shorthand for $[1:n]$. For a vector $\mathbf{a}\in \{0,1\}^n$, we denote by $w(\mathbf{a})$ its Hamming weight (or, simply, weight), i.e., $w(\mathbf{a})$ is the number of ones present in $\mathbf{a}$, and we use the notation supp$(\mathbf{a})$ to denote the set of coordinates $\{i\in[n]:\ a_i = 1\}$. We say that a length-$n$ vector $\mathbf{a}$ is supported on $S\subseteq [n]$, if supp$(\mathbf{a})\subseteq S$. Given a pair of vectors $\mathbf{a},\mathbf{b} \in \mathbb{R}^n$, we define the Hamming distance (or simply, distance) $d(\mathbf{a},\mathbf{b})$ between $\mathbf{a}$ and $\mathbf{b}$ as the number of coordinates in which they differ, i.e., $d(\mathbf{a},\mathbf{b}) = |\{i\in [n]:\ a_i\neq b_i\}|$. Further, given a length-$n$ vector $\mathbf{a}$, let $a_i^j$ denote the vector $(a_k:\ i\leq k\leq j)$, for $j\leq n$. Given a length $n$, we let $\mathbf{B}_n(i)$ denote the binary representation of $0\leq i\leq 2^n-1$. The subscript `$n$' will be dropped when evident from the context. Further, for $\mathbf{z}\in \{0,1\}^n$, we define $w_\mathbf{z}(\mathbf{a})$ as $w_\mathbf{z}(\mathbf{a}):=|\text{supp}(\mathbf{z})\cap \text{supp}(\mathbf{a})|$. For $r\in \{0,1\}$, we define the vector $r^m$ to be the length-$m$ vector all of whose symbols equal $r$. Further, given two binary vectors $\mathbf{a} = (a_1,\ldots,a_n)$ and $\mathbf{b} = (b_1,\ldots,b_n)$, their sum (over $\mathbb{F}_2$) is $\mathbf{a}+\mathbf{b}:=(a_1+b_1,\ldots,a_n+b_n)$, where all the summations are over $\mathbb{F}_2$, their dot product $\mathbf{a}\cdot \mathbf{b}$ equals $\sum_{i=1}^n a_i b_i$, and their concatenation is denoted by $\mathbf{a}\mathbf{b}$ or $(\mathbf{a}\mid \mathbf{b})$. 
For a given real number $r\in \mathbb{R}$, we define $\left \lfloor r\right \rfloor$ and $\left \lceil r\right \rceil$ to be, respectively, the largest integer less than or equal to $r$, and the smallest integer larger than or equal to $r$. 
For a set $P\subseteq \{0,1\}^n$, we use the notation $$\mathds{1}_{P}(\mathbf{x}) = \begin{cases}1,\ \text{if $\mathbf{x}\in P$},\\0,\ \text{otherwise,}\end{cases}$$
and we write $\mathds{1}_P$ as the indicator function $\left(\mathds{1}_P(\mathbf{x}):\ \mathbf{x}\in \{0,1\}^n\right)$. 
\subsection{Block Codes and Constrained Sequences}
We recall the following definitions of block codes and linear codes over $\mathbb{F}_2$ (see, for example, Chapters 1 and 2 of \cite{roth_coding_theory}).

\begin{definition}
	An $(n,M)$ block code $\mathcal{C}$ over $\mathbb{F}_2$ is a nonempty subset of $\mathbb{F}_2^n$, with $|\mathcal{C}| = M$. The rate of the block code $\mathcal{C}$ is given by
$
	\text{rate}(\mathcal{C}) := \frac{\log_2 M}{n}.
$
\end{definition}
\begin{definition}
The minimum distance $d(\mathcal{C})$ of a block code $\mathcal{C}$ is the minimum distance between any two distinct codewords of $\mathcal{C}$, i.e., 
$
d(\mathcal{C}) = \min\limits_{\mathbf{c}_1,\mathbf{c}_2\in \mathcal{C}:\ \mathbf{c}_1\neq \mathbf{c}_2} d(\mathbf{c}_1,\mathbf{c}_2).
$
\end{definition}
An $(n,M)$ block code with minimum distance $d$ will be called an $(n,M,d)$ block code. 

Moreover, given a sequence of block codes $\{\mathcal{C}^{(n)}\}_{n\geq 1}$, if rate$(\mathcal{C}^{(n)})\xrightarrow{n\to \infty} R$, for some $R\in [0,1]$, then we say that $\{\mathcal{C}^{(n)}\}_{n\geq 1}$ is of rate $R$.
\begin{definition}
	An $[n,k]$ linear code $\mathcal{C}$ over $\mathbb{F}_2$ is an $(n,2^k)$ block code that is a subspace of $\mathbb{F}_2^n$.
\end{definition}
{The dual code $\mathcal{C}^\perp\subseteq \{0,1\}^n$ of an $[n,k]$ linear code $\mathcal{C}$ over $\mathbb{F}_2$ is defined as $\mathcal{C}^\perp = \{\mathbf{v}\in \{0,1\}^n: \mathbf{v}\cdot \mathbf{c} = 0,\text{ for all }\mathbf{c}\in \mathcal{C}\}$, with the dimension of $\mathcal{C}^\perp$ being $n-k$}. A constraint is represented by a set $\mathcal{A}\subseteq \{0,1\}^n$ of binary words. We call the sequences in $\mathcal{A}$ as \emph{constrained} sequences, and refer to a block code $\mathcal{C}$, all of whose codewords lie in $\mathcal{A}$, as a ``constrained code''. We refer the reader to \cite{Roth} for a detailed exposition on constrained sequences and coding in the presence of constraints. Note that we make no further assumption about the constrained system (such as it being finite-type, almost-finite-type, irreducible, etc.). Given such a collection of sets of constrained sequences $\{\mathcal{A}_n\}_{n\geq 1}$ for each blocklength $n\geq 1$, where $\mathcal{A}_n\subseteq \{0,1\}^n$, for all $n$, the noiseless capacity (see Chapter 3 of \cite{Roth}) of the constraint is defined as
$
C_0:= \limsup\limits_{n\to \infty} \frac{\log_2 |\mathcal{A}_n|}{n}.
$

Further, for a given blocklength $n$, we use the notation $A(n,d;\mathcal{A})$ to denote the size of the largest constrained code, of minimum distance at least $d$, such that all of its codewords lie in $\mathcal{A}$. More formally,
$
A(n,d;\mathcal{A}):= \max\limits_{\substack{\mathcal{C}\subseteq \mathcal{A}:\ d(\mathcal{C})\geq d}} |\mathcal{C}|.
$
For the case where $\mathcal{A} = \{0,1\}^n$, we write $A(n,d;\mathcal{A})$ as simply $A(n,d)$.
\subsection{Fourier Expansions of Functions}
\label{sec:fourierprelim}
Consider functions $f: \{0,1\}^n\to \mathbb{R}$, mapping $\mathbf{x} = (x_1,\ldots,x_n) \in \{0,1\}^n$ to $f(\mathbf{x})\in \mathbb{R}$. If the range of $f$ is $\{0,1\}$, then $f$ is called a Boolean function.  Now, given any function $f: \{0,1\}^n\to \mathbb{R}$ and a vector $\mathbf{s} = (s_1,\ldots,s_n)\in \{0,1\}^n$, we define the Fourier coefficient of $f$ at $\mathbf{s}$ as
\[
\widehat{f}(\mathbf{s}):= \frac{1}{2^n}\sum_{\mathbf{x}\in \{0,1\}^n} f(\mathbf{x})\cdot (-1)^{\mathbf{x}\cdot \mathbf{s}}.
\]
The function $\widehat{f}$ is known as the Fourier transform (sometimes called the Hadamard transform) of $f$. Moreover, the functions $\left(\chi_{\mathbf{s}}: \mathbf{s}\in \{0,1\}^n\right)$, where $\chi_{\mathbf{s}}(\mathbf{x}):=(-1)^{\mathbf{x}\cdot \mathbf{s}}$, form a basis for the vector space $V$ of functions $f: \{0,1\}^n\to \mathbb{R}$. If we define an inner product $\langle \cdot, \cdot \rangle$ over the vector space $V$, as follows:
\[
\langle f,g \rangle := \frac{1}{2^n}\sum_{\mathbf{x}\in \{0,1\}^n} f(\mathbf{x}) g(\mathbf{x}),
\]
for functions $f,g \in V$, we also have that the basis functions $\left(\chi_{\mathbf{s}}: \mathbf{s}\in \{0,1\}^n\right)$ are orthonormal, in that
\[
\langle \chi_{\mathbf{s}},\chi_{\mathbf{s}^\prime} \rangle = 
\begin{cases}
	1,\ \text{if $\mathbf{s} = \mathbf{s}^\prime$},\\
	0,\ \text{otherwise.}
\end{cases}
\]
For more details on the Fourier analysis over $\mathbb{F}_2^n$, we refer the reader to \cite{ryanodonnell}. In our paper, we shall make use of Plancherel's Theorem from Fourier analysis, which is recalled below, without proof ({see \cite[Chapter 1, p. 26]{ryanodonnell}}).
\begin{theorem}[Plancherel's Theorem]
	For any $f,g \in \{0,1\}^n\to \mathbb{R}$, we have that
	\[
	\langle f,g \rangle = \sum_{\mathbf{s}\in \{0,1\}^n} \widehat{f}(\mathbf{s})\widehat{g}(\mathbf{s}).
	\]
\end{theorem}
We also recall the operation of convolution of two functions $f,g: \{0,1\}^n\to \mathbb{R}$, defined as
\[
f\star g(\mathbf{x}) = \frac{1}{2^n}\sum_{\mathbf{z}\in \{0,1\}^n} f(\mathbf{z})\cdot g(\mathbf{x}+\mathbf{z}),
\]
where the `$+$' operation in $\mathbf{x}+\mathbf{z}$ above is over vectors in $\mathbb{F}_2^n$. It is well-known (see \cite{ryanodonnell}) that the Fourier transform
$
\widehat{f\star g}(\mathbf{s}) = \widehat{f}(\mathbf{s})\cdot \widehat{g}(\mathbf{s}),
$
for any $\mathbf{s}\in \{0,1\}^n$.

\section{Main Results}
\label{sec:main}
\subsection{Counting Constrained Codewords in Linear Codes}
\label{sec:prelim}
First, we work towards characterizing the number of constrained codewords in an arbitrary linear code. Consider an $[n,k]$ linear code $\mathcal{C}$. Suppose that we are interested in computing the number of codewords $\mathbf{c}\in \mathcal{C}$, each of which satisfies a certain property, which we call a constraint. Let $\mathcal{A}\subseteq \{0,1\}^n$ denote the set of length-$n$ words that respect the constraint. We let $N(\mathcal{C};\mathcal{A})$ denote the number of such constrained codewords in $\mathcal{C}$. We can then write
\begin{align}
	\label{eq:countC}
	N(\mathcal{C};\mathcal{A}) &= \sum_{\mathbf{c}\in \mathcal{C}} \mathds{1}_{\mathcal{A}}(\mathbf{c}).
\end{align}
Observe that the summation in \eqref{eq:countC} is over a set of size $2^k$, which could be quite large, especially when $k>n/2$. Our interest is in obtaining insight into the summation above, by employing a simple trick from the Fourier expansions of Boolean functions. 

\begin{theorem}
	\label{thm:lincount}
	Given a linear code $\mathcal{C}$ of blocklength $n$ and a set $\mathcal{A}\subseteq \{0,1\}^n$, we have that
	\[
	N(\mathcal{C};\mathcal{A}) = \left\lvert \mathcal{C}\right\rvert \cdot \sum_{\mathbf{s}\in \mathcal{C}^{\perp}} \widehat{\mathds{1}_\mathcal{A}}(\mathbf{s}).
	\]
\end{theorem}
\begin{proof}
	The proof is a straightforward application of Plancherel's Theorem. Observe that
	\begin{align}
		N(\mathcal{C};\mathcal{A}) &= \sum_{\mathbf{c}\in \mathcal{C}} \mathds{1}_{\mathcal{A}}(\mathbf{c}) \notag\\
		&= \sum_{\mathbf{x}\in \{0,1\}^n} \mathds{1}_{\mathcal{A}}(\mathbf{x})\cdot \mathds{1}_{\mathcal{C}}(\mathbf{x}) 
= 2^n\cdot \sum_{\mathbf{s}\in \{0,1\}^n} \widehat{\mathds{1}_{\mathcal{A}}}(\mathbf{s})\cdot \widehat{\mathds{1}_{\mathcal{C}}}(\mathbf{s}) \label{eq:inter1}.
	\end{align}
	Now, via arguments similar to Lemma 2 in Chapter 5, {p. 127}, of \cite{mws}, we have that
	\begin{equation}
		\label{eq:Ccoeff}
		\widehat{\mathds{1}_{\mathcal{C}}}(\mathbf{s}) = \begin{cases}
			\frac{\left\lvert \mathcal{C} \right\rvert}{2^n},\ \text{if $\mathbf{s}\in \mathcal{C}^\perp$},\\
			0,\ \text{otherwise}.
		\end{cases}
	\end{equation}
%
%
%
	Plugging \eqref{eq:Ccoeff} back in \eqref{eq:inter1}, we obtain the statement of the theorem.
\end{proof}

Theorem \ref{thm:lincount} provides an alternative approach to addressing our problem of counting constrained codewords in linear codes. In particular, note that if $\mathcal{C}$ had large dimension, i.e., if $k>n/2$, then, it is computationally less intensive to calculate the number of constrained codewords using Theorem \ref{thm:lincount}, provided we knew the Fourier coefficients $\widehat{\mathds{1}_\mathcal{A}}(\mathbf{s})$, since $\text{dim}\left(\mathcal{C}^\perp\right) = n-k<n/2$, in this case. Additionally, if the structure of the Fourier coefficients is simple to handle, we could also use Theorem \ref{thm:lincount} to construct linear codes that have a large (or small) number of constrained codewords, or to obtain estimates of the number of constrained codewords in a fixed linear code. 

In Section I of the supplementary material for this manuscript, we discuss the connection between Theorem \ref{thm:lincount} and the well-known MacWilliams' identities for linear codes \cite{macwilliams}. The material in this section of the supplement is well-known, and we provide it for completeness.

In Section \ref{sec:eg}, we shall look at specific examples of constraints and apply Theorem \ref{thm:lincount} above. In particular, as recurring motifs, we shall consider the $[2^m-1,2^m-1-m]$ binary Hamming code, for $m\geq 1$, and the binary Reed-Muller codes. 
We shall first fix a canonical ordering of coordinates for the codes that we analyze. For the binary Hamming code, we assume that a parity-check matrix $H_{\text{Ham}}$ is such that its $i^\text{th}$ column is $ \mathbf{B}_m(i)$, for $1\leq i\leq 2^m-1$.

Note that the Reed-Muller (RM) family of codes is known to achieve the capacities of BMS channels under bit-MAP decoding \cite{Reeves} (see also \cite{kud1}). Thus, RM codes are linear codes that offer the maximum resilience to symmetric, stochastic noise, for a given rate. For $m\geq 1$ and $r\leq m$, the $r^\text{th}$-order binary Reed-Muller code RM$(m,r)$ is the set of binary vectors obtained as evaluations of multilinear Boolean polynomials $f(x_1,\ldots,x_m)$ in the variables $x_1,\ldots,x_m$, of maximum degree $r$, on points of the unit hypercube (see Chapter 13 of \cite{mws} for more information on Reed-Muller codes). We use the convention that the coordinates of RM$(m,r)$ are written as binary $m$-tuples that are ordered according to the standard lexicographic ordering, i.e., the $i^{\text{th}}$ coordinate from the start is the $m$-tuple $\mathbf{B}_m(i-1)$, for $1\leq i\leq 2^m$. We thus have that the blocklength of RM$(m,r)$ is $n=2^m$ and dim$(\text{RM}(m,r)) = \sum_{i=0}^r {m\choose i}=:{m\choose \leq r}$.

\subsection{A Linear Program for Constrained Systems}
\label{sec:lp}
In this section, we consider the problem of upper bounding the sizes of constrained codes with a prescribed minimum distance. In particular, we present a linear program (LP) to upper bound $A(n,d;\mathcal{A})$, for any $\mathcal{A}\subseteq \{0,1\}^n$. This LP is based on Delsarte's linear programming approach \cite{delsarte} to bounding from above the value of $A(n,d)$, for $n\geq 1$ and $1\leq d\leq n$. We first recall Delsarte's LP\footnote{The version of Delsarte's LP that is most often used in papers in coding theory, such as in \cite{mrrw}, is obtained after symmetrizing $\mathsf{Del}(n,d)$. In particular, the common version of Delsarte's LP is $\mathsf{Del}_{/S_n}(n,d)$ (see the remark following Theorem \ref{thm:equiv}), where $S_n$ is the symmetry group on $n$ elements.}, which we call $\mathsf{Del}(n,d)$. Given an LP $\mathsf{L}$, we denote by \textsc{{OPT}}$(\mathsf{L})$ its optimal value, and for any feasible solution $f$ of $\mathsf{L}$, we denote the value of the objective function of $\mathsf{L}$ evaluated at $f$ as val$_\mathsf{L}(f)$. The subscript will be omitted when the LP being referred to is clear from the context. We remark here that the LPs in this paper can return  non-integral optimal values, and that integer upper bounds on the sizes of codes can be obtained by suitable rounding of real numbers. The LP $\mathsf{Del}(n,d)$ is given below:
		\begin{align}
			&\underset{f:\ \{0,1\}^n\to \mathbb{R}}{\text{maximize}}\quad \sum_{\mathbf{x}\in \{0,1\}^n} f(\mathbf{x}) \tag{Obj}\\
			&\text{subject to:} \notag\\
			&\ \ f(\mathbf{x})\geq 0,\ \forall\ \mathbf{x}\in \{0,1\}^n, \tag{C1}\\
			&\ \ \widehat{f}(\mathbf{s})\geq 0,\ \forall\ \mathbf{s}\in \{0,1\}^n, \tag{C2}\\
			&\ \ f(\mathbf{x}) = 0,\ \text{if $1\leq w(\mathbf{x})\leq d-1$}, \tag{C3}\\
			&\ \ f(0^n) = 1. \tag{C4}
		\end{align}
{We then have the following well-known result (see, for example, Section 3 in \cite{linial1}). We provide a complete proof, since the arguments within lead us to the construction of our LP for constrained systems.}
{
\begin{theorem}
	\label{thm:delnd}
{	The inequality $A(n,d)\leq $ \textsc{{OPT}}$(\text{Del}(n,d))$ holds.}
\end{theorem}
}
\begin{proof}
For any block code $\mathcal{C}$ of blocklength $n$ and minimum distance at least $d$, let $\mathds{1}_\mathcal{C}$ denote its indicator function. 
Let us define
$
f_\mathcal{C} := \frac{2^n}{|\mathcal{C}|}\mathds{1}_\mathcal{C}\star \mathds{1}_\mathcal{C}.
$
We claim that $f_\mathcal{C}$ is a feasible solution for $\mathsf{Del}(n,d)$, with {val$(f_\mathcal{C}) =  |\mathcal{C}|$}. Indeed, observe that (C1) is trivially satisfied, by the definition of the convolution operator. Further, since $\widehat{f_\mathcal{C}} = \frac{2^n}{|\mathcal{C}|}\cdot \widehat{\mathds{1}_\mathcal{C}}^2$, (C2) is satisfied as well. Next, note that (C3) also holds since $\mathcal{C}$ is such that $d(\mathcal{C})\geq d$, then $\mathds{1}_\mathcal{C}(\mathbf{x}+\mathbf{z}) = 0$, for all $\mathbf{z}\in \mathcal{C}$ and any $\mathbf{x}$ such that $1\leq w(\mathbf{x})\leq d-1$. Finally,
\begin{align*}
	f_\mathcal{C}(0^n)&= \frac{1}{|\mathcal{C}|}\sum_{\mathbf{z}\in \{0,1\}^n} \mathds{1}_\mathcal{C}(\mathbf{z})\cdot \mathds{1}_\mathcal{C}(\mathbf{z})\\
	&= \frac{1}{|\mathcal{C}|}\sum_{\mathbf{z}\in \{0,1\}^n} \mathds{1}_\mathcal{C}(\mathbf{z}) = 1,
\end{align*}
thereby satisfying (C4) also. Now, the objective value val$(f_\mathcal{C})$ is given by
\begin{align}
	\sum_{\mathbf{x}\in \{0,1\}^n} f_\mathcal{C}(\mathbf{x})&= \frac{1}{|\mathcal{C}|}\sum_{\mathbf{x}\in \{0,1\}^n} \sum_{\mathbf{z}\in \{0,1\}^n} \mathds{1}_\mathcal{C}(\mathbf{z})\cdot \mathds{1}_\mathcal{C}(\mathbf{x}+\mathbf{z}) \notag\\
	&= \frac{1}{|\mathcal{C}|}\sum_{\mathbf{z}\in \{0,1\}^n}\mathds{1}_\mathcal{C}(\mathbf{z}) \sum_{\mathbf{x}\in \{0,1\}^n} \mathds{1}_\mathcal{C}(\mathbf{x}+\mathbf{z}) = |\mathcal{C}|. \label{eq:temp1}
\end{align}
Hence, it follows that the optimal value of the LP, {$\textsc{{OPT}}(\mathsf{Del}(n,d))\geq |\mathcal{C}|$}, and since this holds for all block codes $\mathcal{C}$ of blocklength $n$ and minimum distance at least $d$, we obtain the statement of the theorem. 
\end{proof}
We refer the reader to \cite{delsarte, mrrw, friedman, sam1, coregliano, linial1} and the references therein for a more detailed treatment of linear programming-based upper bounds on the sizes of block codes and linear codes, and for the derivation of analytical upper bounds via the dual LP or using modern Fourier-theoretic or expander graph-based arguments.

Our LP, which we call $\mathsf{Del}(n,d;\mathcal{A})$, is but a small modification of $\mathsf{Del}(n,d)$, to take into account the fact that all codewords of the code of minimum distance at least $d$, whose size we are attempting to bound, must also lie in the set $\mathcal{A}\subseteq \mathbb{F}_2^n$. The LP $\mathsf{Del}(n,d;\mathcal{A})$ is:
		\begin{align}
			&\underset{f:\ \{0,1\}^n\to \mathbb{R}}{\text{maximize}}\quad \sum_{\mathbf{x}\in \{0,1\}^n} f(\mathbf{x}) \tag{Obj$^\prime$}\\
			&\text{subject to:} \notag\\
			&\ \ f(\mathbf{x})\geq 0,\ \forall\ \mathbf{x}\in \{0,1\}^n, \tag{D1}\\
			&\ \ \widehat{f}(\mathbf{s})\geq 0,\ \forall\ \mathbf{s}\in \{0,1\}^n, \tag{D2}\\
			&\ \ f(\mathbf{x}) = 0,\ \text{if $1\leq w(\mathbf{x})\leq d-1$}, \tag{D3}\\
			&\ \ f(0^n)\leq \textsc{{OPT}}(\mathsf{Del}(n,d)), \tag{D4}\\
			&\ \ f(\mathbf{x})\leq 2^n\cdot (\mathds{1}_\mathcal{A}\star \mathds{1}_\mathcal{A})(\mathbf{x}),\ \forall\ \mathbf{x}\in \{0,1\}^n. \tag{D5}
		\end{align}
Like in the case with $\mathsf{Del}(n,d)$, we have an upper bound on $A(n,d;\mathcal{A})$ via the optimal value of $\mathsf{Del}(n,d;\mathcal{A})$.
{
\begin{theorem}
	\label{thm:delndA}
{	For any $\mathcal{A}\subseteq \{0,1\}^n$, we have $A(n,d;\mathcal{A})\leq \textsc{{OPT}}(\mathsf{Del}(n,d;\mathcal{A}))^{1/2}$.}
\end{theorem}
}
\begin{proof}
The proof is very similar to that of Theorem \ref{thm:delnd}. Let $\mathcal{C}_\mathcal{A}$ be any length-$n$ constrained code, with $d(\mathcal{C}_\mathcal{A})\geq d$, such that all codewords in $\mathcal{C}_\mathcal{A}$ lie in $\mathcal{A}$. Observe that we can write $\mathcal{C}_\mathcal{A}$ as $\mathcal{C}\cap \mathcal{A}$, for some block (not necessarily constrained) code $\mathcal{C}$, with $d(\mathcal{C})\geq d$. Thus, an upper bound on $\max\limits_{\substack{\mathcal{C}:\ d(\mathcal{C})\geq d}}|\mathcal{C}\cap \mathcal{A}|$ serves as an upper bound on (and in fact, equals) $A(n,d;\mathcal{A})$. 

Let $\mathds{1}_\mathcal{C}$ be the indicator function of a block code $\mathcal{C}$ as above, and let $\mathds{1}_\mathcal{A}$ be the indicator function of the constraint. We define
$
f_{\mathcal{C},\mathcal{A}} := 2^n\cdot (\mathds{1}_\mathcal{C}\mathds{1}_\mathcal{A}\star \mathds{1}_\mathcal{C}\mathds{1}_\mathcal{A}),
$
and claim that $f_{\mathcal{C},\mathcal{A}}$ is a feasible solution for $\mathsf{Del}(n,d;\mathcal{A})$, with the objective function (Obj$^\prime$) evaluating to $|\mathcal{C}\cap \mathcal{A}|^2$. To see this, note that the LP constraints (D1)--(D3) are satisfied for the same reasons as why $f_\mathcal{C}$ satisfied (C1)--(C3) in $\mathsf{Del}(n,d)$ (see the proof of Theorem \ref{thm:delnd}). Furthermore,
\begin{align*}
	f_{\mathcal{C},\mathcal{A}}(0^n) &= |\mathcal{C}\cap \mathcal{A}|\leq |\mathcal{C}|\leq \textsc{{OPT}}(\mathsf{Del}(n,d)),
\end{align*}
since $\mathcal{C}$ is a block code of distance at least $d$. Hence, (D4) is satisfied by $f_{\mathcal{C},\mathcal{A}}$. Finally, observe that for any $\mathbf{x}\in \{0,1\}^n$, 
\begin{align*}
	f_{\mathcal{C},\mathcal{A}}(\mathbf{x}) &= \sum_{\mathbf{z}\in \{0,1\}^n} \mathds{1}_{\mathcal{C}}\mathds{1}_{\mathcal{A}}(\mathbf{z})\cdot \mathds{1}_\mathcal{C}(\mathbf{x}+\mathbf{z})\mathds{1}_\mathcal{A}(\mathbf{x}+\mathbf{z})\\
	&\leq \sum_{\mathbf{z}\in \{0,1\}^n} \mathds{1}_{\mathcal{A}}(\mathbf{z})\cdot \mathds{1}_\mathcal{A}(\mathbf{x}+\mathbf{z}) = 2^n\cdot (\mathds{1}_\mathcal{A}\star \mathds{1}_\mathcal{A})(\mathbf{x}),
\end{align*}
showing that (D5) also holds. Now, note that {val$(f_{\mathcal{C},\mathcal{A}})$ = $\sum_{\mathbf{x}\in \{0,1\}^n} f_{\mathcal{C},\mathcal{A}}(\mathbf{x})$$ = |\mathcal{C}\cap \mathcal{A}|^2$}, by calculations as in \eqref{eq:temp1}. Hence, we have that \textsc{{OPT}}$(\mathsf{Del}(n,d;\mathcal{A}))\geq |\mathcal{C}\cap \mathcal{A}|^2$, for any $\mathcal{C}$ with minimum distance at least $d$. The statement of the theorem then follows.
\end{proof}
{In Section IV of the paper, we obtain numerical upper bounds on the sizes of constrained codes with a given minimum Hamming distance, for a number of constraints, via an application Theorem \ref{thm:delndA} (and Theorem \ref{thm:equiv}). We now discuss a couple of observations about $\mathsf{Del}(n,d;\mathcal{A})$, stated as propositions.}
{
\begin{proposition}
	\label{prop:A}
	{For any $\mathcal{A}\subseteq \{0,1\}^n$, the inequality 
	$
	(\textsc{{OPT}}(\mathsf{Del}(n,d;\mathcal{A})))^{1/2}\leq |\mathcal{A}|
	$ holds.}
\end{proposition}
}
\begin{proof}
 Note that by (D5), for any feasible solution $f$ of $\mathsf{Del}(n,d;\mathcal{A})$, the objective value
\begin{align*}
	\sum_{\mathbf{x}\in \{0,1\}^n} f(\mathbf{x})&\leq \sum_{\mathbf{x}\in \{0,1\}^n} \sum_{\mathbf{z}\in \{0,1\}^n} \mathds{1}_\mathcal{A}(\mathbf{z})\cdot \mathds{1}_\mathcal{A}(\mathbf{x}+\mathbf{z})\\
	&= \sum_{\mathbf{z}\in \{0,1\}^n} \mathds{1}_\mathcal{A}(\mathbf{z}) \sum_{\mathbf{x}\in \{0,1\}^n} \mathds{1}_\mathcal{A}(\mathbf{x}+\mathbf{z})
	= |\mathcal{A}|^2.
\end{align*}
The statement of the proposition then follows.
\end{proof}
\begin{proposition}
	\label{prop:valdelub}
	For any $\mathcal{A}\subseteq \{0,1\}^n$, we have 
	$
	(\textsc{{OPT}}(\mathsf{Del}(n,d;\mathcal{A})))^{1/2}\leq \textsc{{OPT}}(\mathsf{Del}(n,d)).
	$
\end{proposition}

\begin{proof}
	Given the LP $\mathsf{Del}(n,d;\mathcal{A})$, defined by the objective function (Obj$^\prime$) and the constraints (D1)--(D5), we define the new LP $\overline{\mathsf{Del}}(n,d)$ with the same objective function (Obj$^\prime$) and using the constraints (D1)--(D4) alone (excluding constraint (D5)). Thus, $\overline{\mathsf{Del}}(n,d)$ is given by:
	\begin{align}
		&\underset{f:\ \{0,1\}^n\to \mathbb{R}}{\text{maximize}}\quad \sum_{\mathbf{x}\in \{0,1\}^n} f(\mathbf{x}) \tag{Obj$^\prime$}\\
		&\text{subject to:} \notag\\
		&\ \ f(\mathbf{x})\geq 0,\ \forall\ \mathbf{x}\in \{0,1\}^n, \tag{D1}\\
		&\ \ \widehat{f}(\mathbf{s})\geq 0,\ \forall\ \mathbf{s}\in \{0,1\}^n, \tag{D2}\\
		&\ \ f(\mathbf{x}) = 0,\ \text{if $1\leq w(\mathbf{x})\leq d-1$}, \tag{D3}\\
		&\ \ f(0^n)\leq \textsc{{OPT}}(\mathsf{Del}(n,d)), \tag{D4}
	\end{align}

	It is therefore clear that for any $\mathcal{A}\subseteq \{0,1\}^n$, we have \textsc{{OPT}}$(\mathsf{Del}(n,d;\mathcal{A}))\leq \textsc{{OPT}}(\overline{\mathsf{Del}}(n,d))$. We now claim that $\textsc{{OPT}}(\overline{\mathsf{Del}}(n,d))= (\textsc{{OPT}}({\mathsf{Del}}(n,d)))^2$. 
	
	First, we shall show that the inequality in constraint (D4) in $\overline{\mathsf{Del}}(n,d)$ can be replaced with an equality. To see this, suppose that $f$ were an optimal solution to $\overline{\mathsf{Del}}(n,d)$, with $f(0^n)<\textsc{{OPT}}(\mathsf{Del}(n,d))$. Let $c>0$ be such that $c\leq \text{val}(\mathsf{Del}(n,d))-f(0^n)$. We then construct the function $\overline{f}: \{0,1\}^n\to \mathbb{R}$ such that $\overline{f}(0^n) = f(0^n)+c$, and $\overline{f}(\mathbf{x}) = f(\mathbf{x})$, for $\mathbf{x}\neq 0^n$. It can then easily be verified that $\overline{f}$ satisfies constraints (D1), (D3) and (D4). Furthermore, $\overline{f}$ satisfies (D2) also, since by linearity of the Fourier transform, for any $\mathbf{s}\in \{0,1\}^n$,
	\begin{align*}
		\widehat{(\overline{f})}(\mathbf{s}) &= \widehat{f}(\mathbf{s})+c\cdot \widehat{\mathds{1}_{\{0^n\}}}(\mathbf{s})\\
		&= \widehat{f}(\mathbf{s})+\frac{c}{2^n}\geq 0.
	\end{align*}
	Hence, $\overline{f}$ is a feasible solution to $\overline{\mathsf{Del}}(n,d)$, with val$(\overline{f}) = \sum\limits_{\mathbf{x}\in \{0,1\}^n} f(\mathbf{x})+c> \text{val}(f)$, which contradicts the optimality of $f$. Hence, any optimal solution to $\overline{\mathsf{Del}}(n,d)$ must be such that (D4) is satisfied with an equality, and we can thus replace the inequality in (D4) with an equality.
	
	Now, in order to prove that $(\textsc{{OPT}}(\overline{\mathsf{Del}}(n,d)))^{1/2} = \textsc{{OPT}}({\mathsf{Del}}(n,d))$, it suffices to observe that any feasible solution $f$ of ${\mathsf{Del}}(n,d)$ yields a feasible solution $\textsc{{OPT}}(\mathsf{Del}(n,d))\cdot f$, to $\overline{\mathsf{Del}}(n,d)$ (with the inequality in (D4) changed to an equality). Likewise, any feasible solution $f$ of $\overline{\mathsf{Del}}(n,d)$ yields a feasible solution $\frac{f}{\textsc{{OPT}}(\mathsf{Del}(n,d))}$, to ${\mathsf{Del}}(n,d)$. Owing to this bijection, we obtain that $\textsc{{OPT}}(\overline{\mathsf{Del}}(n,d)) = \textsc{{OPT}}({\mathsf{Del}}(n,d))^2$.
	
	Using the fact that \textsc{{OPT}}$(\mathsf{Del}(n,d;\mathcal{A}))\leq \textsc{{OPT}}(\overline{\mathsf{Del}}(n,d))$, we obtain the statement of the proposition.
\end{proof}

{From Propositions \ref{prop:A} and \ref{prop:valdelub}}, we obtain that the size of the largest constrained code with minimum distance at least $d$ obeys $A(n,d;\mathcal{A})\leq $ $ \min\{\textsc{{OPT}}(\mathsf{Del}(n,d)),|\mathcal{A}|\}$, for all constraints represented by $\mathcal{A}\subseteq \{0,1\}^n$. 

\subsection{Symmetrizing  $\mathsf{Del}(n,d;\mathcal{A})$}
\label{sec:sym}
The linear program $\mathsf{Del}(n,d;\mathcal{A})$ discussed in Section \ref{sec:lp}, for a fixed $\mathcal{A}\subseteq \mathbb{F}_2^n$, suffers from the drawback that the variables, which are precisely the values $(f(\mathbf{x}):\ \mathbf{x}\in \{0,1\}^n)$, are $2^n$ in number, i.e., exponentially large in the blocklength. The number of LP constraints, similarly, are exponentially large in $n$. It would therefore be of interest to check if the size of the linear program $\mathsf{Del}(n,d;\mathcal{A})$, which is the sum of the number of variables and the number of LP constraints, can be reduced, using symmetries present in the formulation.

Our exposition in this section on symmetrizing $\mathsf{Del}(n,d;\mathcal{A})$, follows that in \cite{linial1} ({see also \cite{margot} for a more general study of symmetrization procedures and \cite{fazeli} for an application to the generalized sphere packing bounds}). Let $S_n$ denote the symmetric group on $n$ elements, which is the set of all permutations $\sigma: [n]\to [n]$. Note that given a length-$n$ vector $\mathbf{x} = (x_1,\ldots,x_n)\in \{0,1\}^n$, a permutation $\sigma\in S_n$ acts on $\mathbf{x}$ as follows:
$
\sigma \cdot \mathbf{x} = (x_{\sigma(1)},x_{\sigma(2)},\ldots,x_{\sigma(n)}).
$
The permutation $\sigma$ also acts on functions $f: \{0,1\}^n\to \mathbb{R}$ via the mapping $(\sigma \circ f) (\mathbf{x})= f(\sigma\cdot \mathbf{x})$, for $\mathbf{x}\in \{0,1\}^n$. Now, given any set $\mathcal{A}\subseteq \mathbb{F}_2^n$, we define the ``symmetry group'' of the constraint represented by $\mathcal{A}$ to be the set of all permutations $\pi \in S_n$ that leave the indicator function $\mathds{1}_\mathcal{A}$ invariant. In other words, the symmetry group $G_\mathcal{A}$ of the constraint represented by $\mathcal{A}$ is the set of all permutations $\pi \in S_n$ such that $\mathds{1}_\mathcal{A} = \pi \circ \mathds{1}_\mathcal{A}$.

Given a group $G\subseteq S_n$ of permutations, which acts on the vectors $\mathbf{x}\in \{0,1\}^n$, we say that $\mathsf{Del}(n,d;\mathcal{A})$ is $G$-invariant, if for all $\sigma \in G$, we have that if $f: \{0,1\}^n\to \mathbb{R}$ is a feasible solution to $\mathsf{Del}(n,d;\mathcal{A})$, then so is $\sigma \circ f$, with val$(f) = $ val$(\sigma\circ f)$. The following proposition then holds:
{
\begin{proposition}
	\label{prop:GAinvar}
	{$\mathsf{Del}(n,d;\mathcal{A})$ is $G_\mathcal{A}$-invariant.}
\end{proposition}
}
{Before we prove the above proposition, we shall state and prove a simple lemma.}
\begin{lemma}
	\label{lem:perm}
	For any function $f: \{0,1\}^n \to \mathbb{R}$ and for any permutation $\sigma \in S_n$,
	\[
	\widehat{\sigma \circ f}(\mathbf{s}) = (\sigma \circ \widehat{f})(\mathbf{s}), \quad \text{for all $\mathbf{s}\in \{0,1\}^n$}.
	\]
\end{lemma}
\begin{proof}
	Observe that 
	\begin{align*}
		\widehat{\sigma \circ f}(\mathbf{s})&= \sum_{\mathbf{x}\in \{0,1\}^n} f(\sigma\cdot \mathbf{x})\cdot (-1)^{\mathbf{x}\cdot \mathbf{s}}\\
		&= \sum_{\mathbf{x}\in \{0,1\}^n} f(\sigma\cdot \mathbf{x})\cdot (-1)^{(\sigma\cdot\mathbf{x})\cdot (\sigma\cdot\mathbf{s})}\\
		&=\sum_{\mathbf{x}\in \{0,1\}^n} f(\mathbf{x})\cdot (-1)^{\mathbf{x}\cdot (\sigma\cdot\mathbf{s})}
		=  (\sigma \circ \widehat{f})(\mathbf{s}).
	\end{align*}
\end{proof}
{We shall now prove Proposition \ref{prop:GAinvar}.}
\begin{proof}[Proof of Proposition \ref{prop:GAinvar}]
Let $\pi\in G_\mathcal{A}$ be a permutation in the symmetry group of $\mathcal{A}$ and let $f$ be some feasible solution to $\mathsf{Del}(n,d;\mathcal{A})$. {We first show that $\pi \circ f$ is also a feasible solution to $\mathsf{Del}(n,d;\mathcal{A})$.}

\begin{itemize}
	\item[(D1)] It is clear that if $f(\mathbf{x})\geq 0$, then $f(\pi \cdot \mathbf{x})\geq 0$, for all $\mathbf{x}\in \{0,1\}^n$.
	\item[(D2)] The fact that if $\widehat{f}(\mathbf{s})\geq 0$, then $\widehat{\pi\circ f}(\mathbf{s})\geq 0$, for all $\mathbf{s}\geq 0$, follows directly from Lemma \ref{lem:perm}.
	
	\item[(D3)] Since any permutation in $G_\mathcal{A}$ also lies in $S_n$ and hence preserves the weights of vectors in $\{0,1\}^n$, we have $(\pi\circ f)(\mathbf{x}) = 0$, for all $\mathbf{x}\in \{0,1\}^n$ such that $1\leq w(\mathbf{x})\leq d-1$.
	\item[(D4)] This constraint is also satisfied by $\pi \circ f$, since $\pi(0^n) = 0^n$, for all $\pi\in G_\mathcal{A}$.
	\item[(D5)] Observe that for any $\pi\in G_\mathcal{A}$,
	\begin{align*}
		2^n\cdot (\mathds{1}_\mathcal{A}\star \mathds{1}_\mathcal{A})(\pi\cdot \mathbf{x})&= \sum_{\mathbf{z}\in \{0,1\}^n} \mathds{1}_\mathcal{A}(\mathbf{z})\cdot \mathds{1}_\mathcal{A}(\pi\cdot \mathbf{x}+\mathbf{z})\\
		&= \sum_{\mathbf{z}\in \{0,1\}^n} \mathds{1}_\mathcal{A}(\pi\cdot \mathbf{z})\cdot \mathds{1}_\mathcal{A}(\pi\cdot \mathbf{x}+\pi\cdot \mathbf{z})\\
		&= \sum_{\mathbf{z}\in \{0,1\}^n} \mathds{1}_\mathcal{A}(\pi\cdot \mathbf{z})\cdot \mathds{1}_\mathcal{A}(\pi\cdot (\mathbf{x}+ \mathbf{z}))\\
		&= \sum_{\mathbf{z}\in \{0,1\}^n} \mathds{1}_\mathcal{A}(\mathbf{z})\cdot \mathds{1}_\mathcal{A}(\mathbf{x}+ \mathbf{z}) = 2^n\cdot (\mathds{1}_\mathcal{A}\star \mathds{1}_\mathcal{A})(\mathbf{x}).
	\end{align*}
	Hence, since for all $\mathbf{x}\in \{0,1\}^n$, we have that
	\begin{align*}
		\pi\circ f(\mathbf{x})&\leq 2^n\cdot (\mathds{1}_\mathcal{A}\star \mathds{1}_\mathcal{A})(\pi\cdot \mathbf{x})\\
		&= 2^n\cdot (\mathds{1}_\mathcal{A}\star \mathds{1}_\mathcal{A})(\mathbf{x}),
	\end{align*}
	{where the equality holds since $\pi \in G_\mathcal{A}$}, it follows that (D5) is also satisfied by $\pi\circ f$.
\end{itemize}
{Finally, we show that the values of the feasible solutions $f$ and $\pi \circ f$ are identical:}
\begin{itemize}
	\item[(Obj$^\prime$)] It is clear that $\sum_{\mathbf{x}} f(\mathbf{x}) = \sum_{\mathbf{x}} f(\pi\cdot \mathbf{x})$, and hence that val$(f) = \text{val}(\pi\circ f)$.
\end{itemize}
\end{proof}
From the preceding discussion, we see that given a feasible solution $f$ to $\mathsf{Del}(n,d;\mathcal{A})$, we can construct the function
$
\overline{f}:= \frac{1}{|G_\mathcal{A}|}\sum\limits_{\pi \in G_{\mathcal{A}}} \pi\circ f,
$
such that $\overline{f}$ is also a feasible solution to the LP (by linearity), with val$(\overline{f}) = \text{val}(f)$. Observe, in addition, that $\overline{f}$ is such that $\pi\circ \overline{f} = \overline{f}$, for all $\pi\in G_\mathcal{A}$. Now, given a group $H$ of permutations of $n$ elements, we define the equivalence relation `$\sim_H$' as follows: for vectors $\mathbf{x},\mathbf{y}\in \{0,1\}^n$, we say that $\mathbf{x}\sim_H \mathbf{y}$, if $\mathbf{y} = \sigma\cdot \mathbf{x}$, for some $\sigma\in H$. Further, we define the set $\{0,1\}^n/H$ to be the collection of equivalence classes under $\sim_H$, or orbits, given the group $H$. From the above discussion, it follows that in order to arrive at an optimal solution to $\mathsf{Del}(n,d;\mathcal{A})$, one can restrict oneself to searching among feasible solutions $f$ that are constant on each orbit $O$ in $\{0,1\}^n /G_\mathcal{A}$. Such functions $f$ can be expressed as
\begin{equation}
	\label{eq:orbits}
	f(\mathbf{x}) = \sum_{O\in \{0,1\}^n/G_\mathcal{A}} a_O\cdot \mathds{1}_O(\mathbf{x}),
\end{equation}
where $a_O \in \mathbb{R}$, for all $O\in \{0,1\}^n/G_\mathcal{A}$. Before we work on symmetrizing the constraints of $\mathsf{Del}(n,d;\mathcal{A})$, we introduce some notation. For an orbit $O\in \{0,1\}^n/G_\mathcal{A}$, we denote by $|O|$ the number of elements in the orbit and by $\mathbf{x}_O$ (or $\mathbf{s}_O$) a representative element of the orbit. Further, for a given element $\mathbf{x}\in \{0,1\}^n$, we define $O(\mathbf{x})$ to be the orbit in which $\mathbf{x}$ lies. We shall now formulate (D1)--(D5) and the objective function (Obj$^\prime$) in $\mathsf{Del}(n,d;\mathcal{A})$, based on \eqref{eq:orbits}.
\begin{itemize}
	\item[(D1$^\prime$)] The fact that $f(\mathbf{x})\geq 0$ for all $\mathbf{x}$ implies that $a_O\geq 0$, for all $O\in \{0,1\}^n/G_\mathcal{A}$.
	\item[(D2$^\prime$)] By the linearity of the Fourier transform operation, we obtain that 
	\[
	\widehat{f}(\mathbf{s}) = \sum_{O\in \{0,1\}^n/G_\mathcal{A}} a_O\cdot \widehat{\mathds{1}_O}(\mathbf{s})\geq 0,
	\]
	for all $\mathbf{s}\in \{0,1\}^n$.
	
	In fact, note that since $G_\mathcal{A} \subseteq S_n$, it can be argued using Lemma \ref{lem:perm} that the above inequality only needs to hold for orbit representatives $\mathbf{s}_O\in \{0,1\}^n$, of $O\in \{0,1\}^n/G_\mathcal{A}$. Indeed, we have that for any $\pi \in G_\mathcal{A}$, and for functions $f$ as in \eqref{eq:orbits},
	\begin{align*}
		\widehat{f}(\pi\cdot \mathbf{s}) &= \widehat{\pi \circ f}(\mathbf{s})= \widehat{f}(\mathbf{s}), 
	\end{align*}
	where the first equality holds by Lemma \ref{lem:perm} and the second holds since $\pi\circ f = f$.
	\item[(D3$^\prime$)] The constraint (D3) implies that
$
	a_O = 0,\text{ for all $O$ such that $1\leq w(\mathbf{x}_O)\leq d-1$},
$
	where $\mathbf{x}_O\in \{0,1\}^n$ is a representative element of the orbit $O\in \{0,1\}^n/G_\mathcal{A}$.
	\item[(D4$^\prime$)] The constraint (D4) becomes:
$
	a_{O(0^n)} {= a_{0^n}}\leq \textsc{{OPT}}(\mathsf{Del}(n,d)),
$
	where $O(0^n) {= \{0^n\}}$ is the orbit that contains the all-zeros word $0^n$.
	\item[(D5$^\prime$)] Similarly, the constraint (D5) reduces to:
$
	a_O\leq 2^n\cdot(\mathds{1}_\mathcal{A}\star \mathds{1}_\mathcal{A})(\mathbf{x}_O),
$
	where, again, $\mathbf{x}_O$ is some representative element of the orbit $O\in \{0,1\}^n/G_\mathcal{A}$.
	\item[(Obj$^{\prime\prime}$)] From \eqref{eq:orbits}, we see that the new objective function simply becomes
	\[
	\underset{a_O\in \mathbb{R}}{\text{maximize}} \sum_{O\in \{0,1\}^n/G_\mathcal{A}} |O|\cdot a_O.
	\]
\end{itemize}
We call the symmetrized version of $\mathsf{Del}(n,d;\mathcal{A})$ as $\mathsf{Del}_{/G_\mathcal{A}}(n,d;\mathcal{A})$, which is given below.
		\begin{align}
			&\underset{\{a_O\in \mathbb{R}:\ O\in \{0,1\}^n/G_\mathcal{A}\}}{\text{maximize}} \ \ \sum_{O} |O|\cdot a_O \tag{Obj$^{\prime\prime}$}\\
			&\text{subject to:} \notag\\
			&\ \ a_O\geq 0, \forall\ O\in \{0,1\}^n/G_\mathcal{A}, \tag{D1$^\prime$}\\
			&\ \ \sum_{O\in \{0,1\}^n/G_\mathcal{A}} a_O\cdot \widehat{\mathds{1}_O}(\mathbf{s}_{\tilde{O}})\geq 0,\ \forall\ \text{orbit rep. } \mathbf{s}_{\tilde{O}}\in \{0,1\}^n, \tag{D2$^\prime$}\\
			&\ \ a_O = 0,\ \text{if $1\leq w(\mathbf{x}_O)\leq d-1$}, \tag{D3$^\prime$}\\
			&\ \ {a_{0^n}}\leq \textsc{{OPT}}(\mathsf{Del}(n,d)), \tag{D4$^\prime$}\\
			&\ \ a_O\leq 2^n\cdot (\mathds{1}_\mathcal{A}\star \mathds{1}_\mathcal{A})(\mathbf{x}_O),\ \forall\ O\in \{0,1\}^n/G_\mathcal{A}. \tag{D5$^\prime$}
		\end{align}
The preceding discussion can then be summarized as a theorem.
\begin{theorem}
	\label{thm:equiv}
	The LPs $\mathsf{Del}(n,d;\mathcal{A})$ and $\mathsf{Del}_{/G_\mathcal{A}}(n,d;\mathcal{A})$ are equivalent in that
	\[
	\textsc{{OPT}}(\mathsf{Del}(n,d;\mathcal{A})) = \textsc{{OPT}}(\mathsf{Del}_{/G_\mathcal{A}}(n,d;\mathcal{A})).
	\]
\end{theorem}

\begin{remark}
	All the above arguments remain valid if we use a subgroup $H$ of the symmetry group $G_\mathcal{A}$ as well. For the special case when $\mathcal{A} = \{0,1\}^n$, we have $G_\mathcal{A} = S_n$, and we then recover the more common version of Delsarte's LP that is $M_\text{LP}(n,d)$ in \cite{mrrw}. It is this version that we use for evaluating the right-hand side of constraints (D4) and (D4$^\prime$), in our numerical examples.
\end{remark}
Observe that in the symmetrized LP $\mathsf{Del}_{/G_\mathcal{A}}(n,d;\mathcal{A})$, the number of variables is the number $N_\mathcal{A}$ of orbits $O\in \{0,1\}^n/G_\mathcal{A}$ and the number of constraints is at most $4N_\mathcal{A}+1$. Hence, if the constraint is such that the number of orbits $N_\mathcal{A}$ induced by its symmetry group is small (as a function of the blocklength $n$), then the size of the symmetrized LP is small. {In the section that follows}, we shall explicitly write down $\mathsf{Del}_{/G_\mathcal{A}}(n,d;\mathcal{A})$, for select constraints (or sets $\mathcal{A}$), and provide numerical results obtained by running $\mathsf{Del}_{/G_\mathcal{A}}(n,d;\mathcal{A})$ on those constraints.
\section{Examples}
We now take up specific examples of constrained sequences and apply Theorem \ref{thm:lincount} and the LP discussed in Section \ref{sec:main} {(see Theorems \ref{thm:delndA} and \ref{thm:equiv}). We wish to highlight the fact that an application of Theorem \ref{thm:lincount} depends on a study of the Fourier transforms of the indicator function of the constraint, which we shall pursue here.}
\label{sec:eg}
\subsection{$2$-Charge Constraint}
\label{sec:2charge}
In this subsection, we work with a special kind of a spectral null constraint \cite{spectralnull1}, \cite{spectralnull2}. This constraint that we shall study is the so-called $2$-charge constraint (see Section 1.5.4 in \cite{Roth}), whose sequences have a spectral null at zero frequency (such a constraint is also called a DC-free constraint). The $2$-charge constraint admits only sequences $\mathbf{y} \in \{-1,+1\}^n$, whose running sum $\sum_{i=1}^{r}y_i$, for any $1\leq r\leq n$, obeys $0\leq \sum_{i=1}^{r}y_i\leq 2$. 

To any sequence $\mathbf{x}\in \{0,1\}^n$, we associate (in a one-one manner) the sequence $\mathbf{y} = ((-1)^{x_1},\ldots,(-1)^{x_n})$ $\in \{-1,+1\}^n$. We let $S_2$ denote the set of sequences $\mathbf{x}\in \{0,1\}^n$ such that $\mathbf{y} = ((-1)^{x_1},\ldots,(-1)^{x_n})$ is $2$-charge constrained. Thus, the set of constrained sequences of interest to us is $\mathcal{A} = S_2$. Figure \ref{fig:S2} shows a state transition graph for sequences in the set $S_2$, in that the binary sequences that lie in $S_2$ can be read off the labels of edges in the graph:
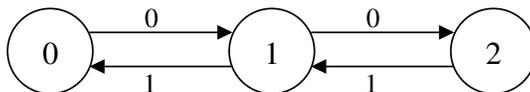
\begin{figure}[!h]
	
	\centering
	\resizebox{0.45\textwidth}{!}{
		\tikzset{every picture/.style={line width=0.75pt}} 
		
		\begin{tikzpicture}[x=0.75pt,y=0.75pt,yscale=-1,xscale=1]
			
			\draw   (320,142) .. controls (320,128.19) and (331.19,117) .. (345,117) .. controls (358.81,117) and (370,128.19) .. (370,142) .. controls (370,155.81) and (358.81,167) .. (345,167) .. controls (331.19,167) and (320,155.81) .. (320,142) -- cycle ;
			\draw   (450,142) .. controls (450,128.19) and (461.19,117) .. (475,117) .. controls (488.81,117) and (500,128.19) .. (500,142) .. controls (500,155.81) and (488.81,167) .. (475,167) .. controls (461.19,167) and (450,155.81) .. (450,142) -- cycle ;
			\draw   (580,142) .. controls (580,128.19) and (591.19,117) .. (605,117) .. controls (618.81,117) and (630,128.19) .. (630,142) .. controls (630,155.81) and (618.81,167) .. (605,167) .. controls (591.19,167) and (580,155.81) .. (580,142) -- cycle ;
			\draw    (368,131) -- (449,131) ;
			\draw [shift={(452,131)}, rotate = 180] [fill={rgb, 255:red, 0; green, 0; blue, 0 }  ][line width=0.08]  [draw opacity=0] (8.93,-4.29) -- (0,0) -- (8.93,4.29) -- cycle    ;
			\draw    (498,131) -- (579,131) ;
			\draw [shift={(582,131)}, rotate = 180] [fill={rgb, 255:red, 0; green, 0; blue, 0 }  ][line width=0.08]  [draw opacity=0] (8.93,-4.29) -- (0,0) -- (8.93,4.29) -- cycle    ;
			\draw    (501,151) -- (582,151) ;
			\draw [shift={(498,151)}, rotate = 0] [fill={rgb, 255:red, 0; green, 0; blue, 0 }  ][line width=0.08]  [draw opacity=0] (8.93,-4.29) -- (0,0) -- (8.93,4.29) -- cycle    ;
			\draw    (371,151) -- (452,151) ;
			\draw [shift={(368,151)}, rotate = 0] [fill={rgb, 255:red, 0; green, 0; blue, 0 }  ][line width=0.08]  [draw opacity=0] (8.93,-4.29) -- (0,0) -- (8.93,4.29) -- cycle    ;
			
			\draw (339,135.4) node [anchor=north west][inner sep=0.75pt]  [font=\large]  {$0$};
			\draw (469,135.4) node [anchor=north west][inner sep=0.75pt]  [font=\large]  {$1$};
			\draw (599,135.4) node [anchor=north west][inner sep=0.75pt]  [font=\large]  {$2$};
			\draw (399,115.4) node [anchor=north west][inner sep=0.75pt]    {$0$};
			\draw (529,115.4) node [anchor=north west][inner sep=0.75pt]    {$0$};
			\draw (528,154.4) node [anchor=north west][inner sep=0.75pt]    {$1$};
			\draw (398,154.4) node [anchor=north west][inner sep=0.75pt]    {$1$};

		\end{tikzpicture}
	}
	\caption{State transition graph for sequences in the set $S_2$.}
	\label{fig:S2}
\end{figure}

We assume that the initial state is $v_0 = 0$. Since labels of paths in the state transition graph (beginning at state $0$) correspond to binary sequences $\mathbf{x}\in S_2$, we denote by $x_i$ the label of the $i^{\text{th}}$ edge in the path. Observe that $x_1 = 0$, by our choice of initial state. Further, for a given path in the graph, we let $v_i$ denote the $i^{\text{th}}$ state, which is the terminal state of the $i^{\text{th}}$ edge. 

{Further, we claim that $\left\lvert S_2\right \rvert = 2^{\left\lfloor \frac{n}{2}\right \rfloor}$}. To see this, observe that for any $\mathbf{x}\in S_2$, the state $v_{2i-1}$, for any $1\leq i\leq n$, equals $1$. Owing to this fact, the label of the $j^\text{th}$ edge, $x_j$, in any path in the graph $G$ in Figure \ref{fig:S2}, can be either $0$ or $1$, when $j = 2i$, and is fixed to be exactly one of $0$ or $1$, when $j=2i+1$, based on the label of the $(j-1)^\text{th}$ edge, for $1\leq j\leq n$. In particular, it holds that $x_{2i} +x_{2i+1} =1$, for all {$1\leq i\leq \left\lfloor \frac{n-1}{2}\right \rfloor$}, with $x_1$ fixed to be $0$. 

\subsubsection{Computation of Fourier Transform}

We now state a lemma that completely determines the Fourier transform of $\mathds{1}_{S_2}$. But before we do so, we need some more notation: we define the set of vectors $\mathcal{B} = \left\{\mathbf{b}_0,\mathbf{b}_1,\ldots,\mathbf{b}_{\left\lceil \frac{n}{2} \right\rceil - 1}\right\}$, where $\mathbf{b}_0 = 10^{n-1}$ and for $1\leq i\leq \left\lceil \frac{n}{2} \right\rceil - 1$, the vector $\mathbf{b}_i$ is such that $b_{i,j} = 1$, for $j\in \{2i,2i+1\}$, and $b_{i,j} = 0$, otherwise. For example, when $n=5$, we have that $\mathcal{B} = \{10000,01100,00011\}$. Let $V_{\mathcal{B}} = \text{span}(\mathcal{B})$. In what follows, we assume that $n\geq 3$.

\begin{lemma}
	\label{lem:fcoeff2c}
	For $n\geq 3$ and for $\mathbf{a} = \left(a_0,a_1,\ldots,a_{{\left\lceil \frac{n}{2} \right\rceil - 1}}\right)\in \{0,1\}^{\left\lceil \frac{n}{2} \right\rceil}$, consider $\mathbf{s} = \sum\limits_{i=0}^{{\left\lceil \frac{n}{2} \right\rceil - 1}} a_i\cdot \mathbf{b}_i$ (where the summation is over $\mathbb{F}_2^n$). Then,
	\[
	\widehat{\mathds{1}_{S_{2}}}\left(\mathbf{s}\right)= 2^{\left\lfloor \frac{n}{2} \right\rfloor - n}\cdot (-1)^{w(\mathbf{a})-a_0}.
	\]
	Further, for $\mathbf{s}\notin V_{\mathcal{B}}$, we have that $\widehat{\mathds{1}_{S_{2}}}\left(\mathbf{s}\right) = 0$.
\end{lemma}
\begin{proof}
	First, we note that for any $\mathbf{s}\in \{0,1\}^n$, 
	\begin{align}
	\widehat{\mathds{1}_{S_2}}(\mathbf{s}) &= \frac{1}{2^n}\sum_{\mathbf{x}\in \{0,1\}^n} \mathds{1}_{S_2}(\mathbf{x})\cdot (-1)^{\mathbf{x}\cdot \mathbf{s}} \notag\\
	&= 2^{-n}\cdot\left(\#\{\mathbf{x}\in S_2:\ w_{\mathbf{s}}(\mathbf{x})\text{ is even}\} - \#\{\mathbf{x}\in S_2:\ w_{\mathbf{s}}(\mathbf{x})\text{ is odd}\}\right) \label{eq:inter3}.
	\end{align}
Now, for $\mathbf{s} = \mathbf{b}_0$, note that since all words $\mathbf{x}\in S_2$ have $x_1 = 0$, we obtain that $\widehat{\mathds{1}_{S_2}}(\mathbf{b}_0) = 2^{-n}\cdot \left\lvert S_2\right \rvert = 2^{\left\lfloor \frac{n}{2} \right\rfloor - n}$.

Further, recall that since $v_{2i-1} = 1$, for any $1\leq i\leq n$, we have that $x_{2i}+x_{2i+1} = 1$ (over $\mathbb{F}_2$). Hence, we see that for $\mathbf{s} = \mathbf{b}_j$, for $1\leq j\leq \left\lceil \frac{n}{2} \right\rceil-1$, it is true that $\#\{\mathbf{x}\in S_2:\ w_{\mathbf{s}}(\mathbf{x})\text{ is odd}\} = \left\lvert S_2\right \rvert = 2^{\left\lfloor \frac{n}{2}\right \rfloor}$ and $\#\{\mathbf{x}\in S_2:\ w_{\mathbf{s}}(\mathbf{x})\text{ is even}\} = 0$. Substituting in \eqref{eq:inter3}, we get that $\widehat{\mathds{1}_{S_2}}(\mathbf{b}_j) = -2^{\left\lfloor \frac{n}{2} \right\rfloor - n}$ for all $1\leq j\leq \left\lceil \frac{n}{2} \right\rceil-1$. Furthermore, we claim that $\widehat{\mathds{1}_{S_{2}}}(0^n) = 2^{\left\lfloor \frac{n}{2} \right\rfloor - n}$. To see this, note that
\begin{align*}
	\widehat{\mathds{1}_{S_{2}}}(0^n) &{= \frac{1}{2^n}\sum_{\mathbf{x}\in S_2} 1 = \frac{|S_2|}{2^n} = 2^{\left\lfloor \frac{n}{2} \right\rfloor - n}}.
\end{align*}

Now, suppose that for some $\mathbf{s}_1, \mathbf{s}_2 \in V_\mathcal{B}$, we have $\widehat{\mathds{1}_{S_2}}(\mathbf{s}_1) = (-1)^{i_1}\cdot 2^{\left\lfloor \frac{n}{2} \right\rfloor - n}$ and $\widehat{\mathds{1}_{S_2}}(\mathbf{s}_2) = (-1)^{i_2}\cdot 2^{\left\lfloor \frac{n}{2} \right\rfloor - n}$, for some $i_1,i_2\in \{0,1\}$, with $w_{\mathbf{s}_1}(\mathbf{x})$ being even for all $\mathbf{x}\in S_2$, if $i_1 = 0$, and odd otherwise {(similar arguments hold for $w_{\mathbf{s}_2}(\mathbf{x})$). Hence, it can be checked that if $i_1 = i_2$, it holds that $w_{\mathbf{s}_1+\mathbf{s}_2}(\mathbf{x})$ is even, and hence, $\widehat{\mathds{1}_{S_2}}(\mathbf{s}_1+\mathbf{s}_2) = 2^{\left\lfloor \frac{n}{2} \right\rfloor - n} = (-1)^{i_1+i_2}\cdot 2^{\left\lfloor \frac{n}{2} \right\rfloor - n}$, and similarly if $i_1\neq i_2$ as well}. By applying this fact iteratively, and using the expressions for the Fourier coefficients $\widehat{\mathds{1}_{S_2}}(\mathbf{b}_j)$, for $0\leq j\leq \left\lceil \frac{n}{2} \right\rceil-1$, we obtain the first part of the lemma.

To show that $\widehat{\mathds{1}_{S_{2}}}\left(\mathbf{s}\right) = 0$ for $\mathbf{s}\notin V_{\mathcal{B}}$, we use Plancherel's Theorem again. Note that
\begin{align}
	\frac{\left\lvert S_2 \right \rvert}{2^n} &\stackrel{}{=} \frac{1}{2^n}\sum_{\mathbf{x}\in \{0,1\}^n} \mathds{1}_{S_2}(\mathbf{x}) \notag\\
	&\stackrel{{(a)}}{=} \frac{1}{2^n}\sum_{\mathbf{x}\in \{0,1\}^n} \mathds{1}_{S_2}^2(\mathbf{x}) \notag\\
	&\stackrel{{(b)}}{=} \sum_{\mathbf{s}\in \{0,1\}^n} \left(\widehat{\mathds{1}_{S_{2}}}(\mathbf{s})\right)^2
	\stackrel{}{=} \sum_{\mathbf{s}\in V_\mathcal{B}} \left(\widehat{\mathds{1}_{S_{2}}}(\mathbf{s})\right)^2+\sum_{\mathbf{s}\notin V_\mathcal{B}} \left(\widehat{\mathds{1}_{S_{2}}}(\mathbf{s})\right)^2 \label{eq:inter4},
\end{align}
where (a) holds since $\mathds{1}_{S_2}$ is a boolean function, and (b) holds by Plancherel's Theorem. 
However, from the first part of the lemma, we get that 
\begin{align*}
	\sum_{\mathbf{s}\in V_\mathcal{B}} \left(\widehat{\mathds{1}_{S_{2}}}(\mathbf{s})\right)^2 &= \left\lvert V_\mathcal{B} \right \rvert \cdot 2^{2\cdot\left(\left\lfloor \frac{n}{2}\right \rfloor - n\right)}\\
	&\stackrel{(c)}{=} 2^{\left\lceil \frac{n}{2}\right \rceil}\cdot 2^{2\cdot\left(\left\lfloor \frac{n}{2}\right \rfloor - n\right)}
	= 2^{-\left\lceil \frac{n}{2}\right \rceil} = \frac{\left\lvert S_2 \right \rvert}{2^n},
\end{align*}
where equality {(c)} follows from the fact that $\left\lvert V_\mathcal{B} \right \rvert =  2^{\left\lceil \frac{n}{2}\right \rceil}$, since $V_\mathcal{B} = \text{span}(\mathcal{B})$ and the vectors in $\mathcal{B}$ are linearly independent. Hence, plugging back in \eqref{eq:inter4}, we obtain that $\sum_{\mathbf{s}\notin V_\mathcal{B}} \left(\widehat{\mathds{1}_{S_{2}}}(\mathbf{s})\right)^2 = 0$, implying that $\widehat{\mathds{1}_{S_{2}}}(\mathbf{s})= 0$, for all $\mathbf{s}\notin V_\mathcal{B}$.
\end{proof}

\subsubsection{Insights from Lemma \ref{lem:fcoeff2c}}

Lemma \ref{lem:fcoeff2c} informs the construction of linear codes $\mathcal{C}$ that have a large number of codewords $\mathbf{c}\in S_2$. In particular, note that from Theorem \ref{thm:lincount}, we have that
\begin{align}
N(\mathcal{C};S_2) &= |\mathcal{C}|\cdot \sum_{\mathbf{s}\in \mathcal{C}^\perp} \widehat{\mathds{1}_{S_{2}}}(\mathbf{s})\notag\\
&= |\mathcal{C}|\cdot \sum_{\mathbf{s}\in \mathcal{C}^\perp \cap V_{\mathcal{B}}} \widehat{\mathds{1}_{S_{2}}}(\mathbf{s}) \label{eq:inter5},
\end{align}
where, for $\mathbf{s}\in \mathcal{C}^\perp \cap V_{\mathcal{B}}$, with $\mathbf{s} = \sum_{i=0}^{\left \lceil \frac{n}{2} \right\rceil -1} a_i\cdot \mathbf{b}_i$ for some $\mathbf{a} = \left(a_0,a_1,\ldots,a_{{\left\lceil \frac{n}{2} \right\rceil - 1}}\right)\in \{0,1\}^{\left\lceil \frac{n}{2} \right\rceil}$, we have that $\widehat{\mathds{1}_{S_{2}}}(\mathbf{s}) =  2^{\left\lfloor \frac{n}{2} \right\rfloor - n}\cdot \left((-1)^{\sum_{j=1}^{{\left\lceil \frac{n}{2} \right\rceil}-1} a_j}\right)$. Now, suppose that $n\geq 3$ and $\mathcal{C}$ is such that $\mathcal{C}^\perp$ \emph{does not} satisfy the criterion (\textbf{C}) below:
\begin{align*}
	\text{(\textbf{C})}\quad \text{For all $\mathbf{s}\in  \mathcal{C}^\perp \cap V_{\mathcal{B}}$, we have $\widehat{\mathds{1}_{S_{2}}}(\mathbf{s})\geq 0$.}
\end{align*}

If (\textbf{C}) does not hold, then, it implies that for some $\mathbf{s}^\star\in  \mathcal{C}^\perp \cap V_{\mathcal{B}}$, it must be that $\widehat{\mathds{1}_{S_{2}}}(\mathbf{s}^\star)<0$. Hence, following the reasoning in the proof of Lemma \ref{lem:fcoeff2c}, since $\mathcal{C}^\perp \cap V_{\mathcal{B}}$ is a vector space, we have that via the map $\mathbf{s} \mapsto \mathbf{s}+\mathbf{s}^\star$, the number of elements $\mathbf{s}\in \mathcal{C}^\perp \cap V_{\mathcal{B}}$ such that $\widehat{\mathds{1}_{S_{2}}}(\mathbf{s}) <0$ equals the number of elements $\mathbf{s}\in \mathcal{C}^\perp \cap V_{\mathcal{B}}$ such that $\widehat{\mathds{1}_{S_{2}}}(\mathbf{s})>0$. Furthermore, since $\left\lvert \widehat{\mathds{1}_{S_{2}}}(\mathbf{s})\right \rvert  = 2^{\left\lfloor \frac{n}{2} \right\rfloor - n}$, for all $\mathbf{s}\in \mathcal{C}^\perp \cap V_{\mathcal{B}}$, we get from \eqref{eq:inter5} that $N(\mathcal{C};S_2) = 0$, in this case.

Hence, in order to construct linear codes $\mathcal{C}$ such that $N(\mathcal{C},S_2) > 0$, we require that criterion (\textbf{C}) is indeed satisfied by the dual code $\mathcal{C}^\perp$ of $\mathcal{C}$, with $\widehat{\mathds{1}_{S_{2}}}(\mathbf{s}^\star) >0$, for some $\mathbf{s}^\star\in \mathcal{C}^\perp$. With this instruction in mind, we can construct linear codes $\mathcal{C}$ such that its dual code $\mathcal{C}^\perp$ contains $t$ linearly independent vectors $\left(\mathbf{s}_1,\ldots,\mathbf{s}_t\right)$ with $\widehat{\mathds{1}_{S_{2}}}(\mathbf{s}_i)>0$, for all $1\leq i\leq t$, and no vectors $\mathbf{s}\in V_{\mathcal{B}}$ with $\widehat{\mathds{1}_{S_{2}}}(\mathbf{s})<0$. In such a case, we obtain that
$
N(\mathcal{C};S_2) = |\mathcal{C}|\cdot 2^{t+\left\lfloor \frac{n}{2} \right\rfloor - n}.
$
From the structure of $V_\mathcal{B}$, we see that the largest number of vectors $\mathbf{s}\in \{0,1\}^n$ such that $\widehat{\mathds{1}_{S_{2}}}(\mathbf{s})>0$, equals $\frac{\left\lvert V_\mathcal{B}\right\rvert}{2} = 2^{\left\lceil \frac{n}{2} \right\rceil - 1}$. Hence, the largest number of linearly independent vectors $t$ as above, is $\left\lceil \frac{n}{2} \right\rceil - 1$.
The discussion above is summarized below as a lemma.
\begin{lemma}
	\label{lem:criterion2c}
	For any linear code $\mathcal{C}$ of blocklength $n\geq 3$, the following are true:
	\begin{enumerate}
		\item If criterion (\textbf{C}) is not satisfied, then, $N(\mathcal{C},S_2) = 0$.
		\item If criterion (\textbf{C}) is satisfied and there exist $t_n\in \left[1:\left\lceil \frac{n}{2} \right\rceil - 1\right]$ linearly independent vectors $\left(\mathbf{s}_1,\ldots,\mathbf{s}_{t_n}\right)$ in $\mathcal{C}^\perp$ with $\widehat{\mathds{1}_{S_{2}}}(\mathbf{s}_i)>0$, for all $1\leq i\leq t_n$, then, $N(\mathcal{C};S_2) = |\mathcal{C}|\cdot 2^{t_n+\left\lfloor \frac{n}{2} \right\rfloor - n}$.
	\end{enumerate}
\end{lemma}

We thus understand that given a linear code whose dual code satisfies item 2 of Lemma \ref{lem:criterion2c}, the rate of the largest constrained subcode, $\mathcal{C}_2$, of $\mathcal{C}$, all of whose codewords are in $S_2$, obeys
\begin{align*}
\text{rate}\left(\mathcal{C}_2\right) &= \frac{\log_2 N(\mathcal{C};S_2)}{n}
= \frac{\log_2\left(|\mathcal{C}|\right)}{n}+\frac{t_n+\left\lfloor \frac{n}{2} \right\rfloor - n}{n}.
\end{align*}
In particular, given a sequence of linear codes $\left\{\mathcal{C}^{(n)}\right\}_{n\geq 1}$ satisfying item 2 of Lemma \ref{lem:criterion2c}, if rate$(\mathcal{C}^{(n)})\xrightarrow{n\to \infty} R\in (0,1)$, then, the rate of their largest constrained subcodes $\left\{\mathcal{C}_2^{(n)}\right\}_{n\geq 1}$, all of whose codewords are in $S_2$, obeys
\begin{equation}
	\label{eq:rate2c}
\liminf_{n\to \infty}\text{rate}\left(\mathcal{C}_2^{(n)}\right) = R-\frac12 + \liminf_{n\to \infty} \frac{t_n}{n}.
\end{equation}

{By arguments similar to those in \cite{pvk}}, we obtain that for the constraint identified by the set $S_2$, there exist cosets of the linear codes $\left\{\mathcal{C}^{(n)}\right\}_{n\geq 1}$ with rate$(\mathcal{C}^{(n)})\xrightarrow{n\to \infty} R$, the rate of the constrained subcodes of which (in the limit as the blocklength goes to infinity) is at least $R-\frac12$. From \eqref{eq:rate2c}, since $t_n\in \left[1:\left\lceil \frac{n}{2} \right\rceil - 1\right]$, we see that we can construct a sequence of linear codes whose {$2$-charge} constrained subcodes are of rate larger than or equal to the coset-averaging lower bound in \cite{pvk}. In other words, it is possible to achieve the coset-averaging rate lower bound {for the $2$-charge constraint} (and potentially more) by using the linear code itself, instead of one of its cosets. 

Specifically, suppose that we choose $t_n = \left\lceil \frac{n}{2}\right\rceil - p_n$, for some positive integer $p_n$ such that $\lim_{n\to \infty} \frac{p_n}{n} = 0$, thereby making dim$\left(\mathcal{C}_n^\perp\right)\geq \frac{\left\lceil \frac{n}{2}\right\rceil - p_n}{n}$, where $\mathcal{C}_n^\perp$ is the dual code of $\mathcal{C}_n$. Note that this implies that $1-R = \lim_{n\to \infty} \text{rate}\left(\mathcal{C}^\perp\right)\geq \frac12$, and hence that $R\in (0,\frac12]$. In this case, by plugging into \eqref{eq:rate2c}, we obtain that the rate of the largest constrained subcodes $\left\{\mathcal{C}_2^{(n)}\right\}_{n\geq 1}$ of $\left\{\mathcal{C}^{(n)}\right\}_{n\geq 1}$ is 
\begin{align*}
	\lim_{n\to \infty}\text{rate}\left(\mathcal{C}_2^{(n)}\right) &= R-\frac12 + \lim_{n\to \infty} \frac{t_n}{n}= R.
\end{align*}
In other words, in the case where $t_n = \left\lceil \frac{n}{2}\right\rceil - p_n$, for $p_n>0$ as above, the asymptotic rate of the codewords that lie in $S_2$ equals the asymptotic rate $R\in (0,\frac12]$ of the code itself.

\subsubsection{Application to Specific Linear Codes}

Next, we shall make use of Theorem \ref{thm:lincount} to compute the number of codewords of specific linear codes $\mathcal{C}$, which lie in $S_2$. 
First, we shall apply our results to the $[2^m-1,2^m-1-m]$ binary Hamming code, for $m\geq 3$. We shall use the coordinate ordering discussed in Section \ref{sec:prelim}. {The proof of the corollary below is provided in Section II of the supplementary material.}

\begin{corollary}
	\label{cor:s2hamming}
	For $m\geq 3$ and for $\mathcal{C}$ being the $[2^m-1,2^m-1-m]$ Hamming code, we have that $N(\mathcal{C};S_2) = 2^{\left\lfloor \frac{2^m-1}{2}\right \rfloor-1}$.
\end{corollary}

Note that in Corollary \ref{cor:s2hamming}, the number of constrained codewords in the linear codes {is} half the total number of constrained codewords, $2^{\left\lfloor \frac{n}{2}\right \rfloor}$, of the same blocklength $n$ as the codes under consideration. However, in the limit as the blocklength goes to infinity, the rates of the subcodes of the single parity-check and Hamming codes that lie in $S_2$, equal the noiseless capacity $C_0$ of the constraint, which in turn equals $\frac12$. 
 
 We then move on to counting constrained codewords in the Reed-Muller (RM) family of codes.
 Using the structure of Fourier coefficients given in Lemma \ref{lem:fcoeff2c} and using the fact that the dual code of RM$(m,r)$ is the code RM$(m,m-r-1)$, for $r\leq m-1$, we numerically calculate the number of constrained codewords $N(\text{RM}(m,r);S_2)$, for certain (large) values of $m$ and $r$. Our results are documented in Table~\ref{tab:rm}. Note that the computational technique in Theorem \ref{thm:lincount} proves particularly useful when the rate of RM$(m,r)$ is larger than $\frac12$, or equivalently, when $r>\left\lceil \frac{m}{2}\right\rceil$. 
 {{The algorithm we have used for generating the entries in  Table~\ref{tab:rm}, as an illustration of the application of Theorem~\ref{thm:lincount}, simply plugs in the Fourier coefficients from Lemma \ref{lem:fcoeff2c}. The time complexity of this algorithm is thus $O(n\cdot |\mathcal{C}^\perp|)$ for $\mathcal{C} = $ RM$(m,r)$.}
 
 {More generally though, observe that the constraint that a word $\mathbf{x}\in \{0,1\}^n$ lies in $S_2$ can be represented by the set of linear equations (over $\mathbb{F}_2$) given by $x_1 = 0$ and $x_{2i} +x_{2i+1} =1$, for all {$1\leq i\leq \left\lfloor \frac{n-1}{2}\right \rfloor$}. Furthermore, any linear code $\mathcal{C}$ of dimension $k$ and parity-check matrix $H$ is such that its codewords $\mathbf{c}$ are solutions to $H\cdot \mathbf{c}^T = 0^{n-k}$. The $2$-charge  constrained codewords in $\mathcal{C}$ can thus be represented as solutions to a system of linear equations over $\mathbb{F}_2$, and the number of such solutions can be determined by Gaussian elimination, in time that is polynomial in the blocklength $n$.}

\begin{table}[t!]
	\centering
	\begin{tabular}{||c || c||c||c||c||c||c||} 
		\hline
		$(m,r)$ & $(4,2)$ &$(4,3)$ &$(5,3)$ &$(6,4)$ &$(7,5)$ &$(8,6)$\\
		\hline
		$N(\text{RM}(m,r);S_2)$& $16$& $128$& $2048$& $6.711\times 10^7$& $1.441\times 10^{17}$& $1.329\times 10^{36}$\\
		\hline
	\end{tabular}
\caption{Table of values of $N(\text{RM}(m,r);S_2)$, for select parameters $m$ and $r$}
\label{tab:rm}
\end{table}

\subsubsection{LP-Based Upper Bounds on the Sizes of $2$-Charge Constrained Codes}
Now, we shall work towards obtaining bounds on the sizes of constrained codes that are subsets of $S_2$, of minimum distance at least $d$. In other words, we are interested in formulating the symmetrized LP $\mathsf{Del}_{/G_{S_2}}(n,d;S_2)$. In what follows, we fix the blocklength $n$ to be odd. Slight modifications of the construction of the symmetry group $G_{S_2}$ and the identification of the orbits, below, yield  $\mathsf{Del}_{/G_{S_2}}(n,d;S_2)$, when $n$ is even.

Now, consider the following permutations, where $n$ is odd:
\begin{enumerate}
	\item For even indices $i\in [n]$, define $\pi_i^{\text{adj}}: [n]\to [n]$, such that $\pi_i^{\text{adj}}(i) = i+1$, $\pi_i^{\text{adj}}(i+1) = i$, and $\pi_i^{\text{adj}}(j) = j$, for $j\notin \{i,i+1\}$.
	
	In words, $\pi_{i,j}^{\text{adj}}$ swaps adjacent positions $i$ and $i+1$, for even $i\in [n]$, and leaves other positions unchanged. 
	\item For even indices $i,j\in [n]$, define ${\pi_{i,j}^{\text{swap}}}: [n]\to [n]$, such that ${\pi_{i,j}^{\text{swap}}}(i) = j$, ${\pi_{i,j}^{\text{swap}}}(i+1) = j+1$, and ${\pi_{i,j}}^{\text{swap}}(j) = i$, ${\pi_{i,j}^{\text{swap}}}(j+1) = i+1$, with ${\pi_{i,j}^{\text{swap}}}(k) = k$, for $k\notin \{i,i+1,j,j+1\}$.
	
	In words, $\pi_{i,j}^{\text{swap}}$ swaps $i$ and $j$, and $i+1$ and $j+1$, for $i,j$ being even, and leaves other positions unchanged. 
\end{enumerate}
 
The discussion above on the sequences in $S_2$ implies that the symmetry group $G_{S_2}$ of the constraint is generated (via compositions) by $\{\pi_i^{\text{adj}}:\ i\text{ even}\}\cup \{{\pi_{i,j}^{\text{swap}}}:\ i,j\text{ even}\}$. Further, consider tuples $\boldsymbol{\alpha}\in \{0,1\}\times \left[0:\left \lfloor \frac{n}{2}\right \rfloor\right] \times \left[0:\left \lfloor \frac{n}{2}\right \rfloor\right]$ of the form $\boldsymbol{\alpha} = (b,t_{00},t_{11})$, with $t_{00}+t_{11}\leq \left \lfloor \frac{n}{2}\right \rfloor$. For a sequence $\mathbf{x}\in \{0,1\}^n$, we identify $b\in \{0,1\}$ with $x_1$, the integer $t_{00}$ with $|\{i:\ i \text{ even and }(x_i,x_{i+1})=(0,0)\}|$, and the integer $t_{11}$ with $|\{i:\ i \text{ even and }(x_i,x_{i+1})=(1,1)\}|$. Note that then $|\{i:\ i \text{ even and }(x_i,x_{i+1})=(0,1)\text{ or }(1,0)\}| = \left \lfloor \frac{n}{2}\right \rfloor - t_{00} - t_{11}$. We thus have that the orbits of the symmetry group of the constraint $\{0,1\}^n/G_{S_2}$ are in one-one correspondence with tuples of the form $\boldsymbol{\alpha} = (b,t_{00},t_{11})$. Observe that the number of orbits is hence bounded above by $2\cdot \left \lceil \frac{n}{2}\right \rceil^2$, and therefore the number of variables and the number of constraints in the LP $\mathsf{Del}_{/G_{S_2}}(n,d;S_2)$, are bounded above by a polynomial function of the blocklength $n$, unlike the number of variables in $\mathsf{Del}(n,d;S_2)$, which equals $2^n$.

Table \ref{tab:2c} shows numerical evaluations of $\mathsf{Del}_{/G_{S_2}}(n,d;S_2)$, when $n=13$, for varying values of $d$. The table also includes comparisons with upper bounds via the generalized sphere packing bound of \cite{cullina} and \cite{fazeli} and with $\mathsf{Del}(n,d)$. We observe that our LP provides tighter upper bounds than those obtained by the sphere packing approach\footnote{{In this table and in all others where the generalized sphere packing bounds are computed, since we only show the bounds for the purpose of comparisons, we have used the LP given in Section II of \cite{fazeli} and not the symmetrized version in Corollary 8 of \cite{fazeli}.}}.

\begin{table}[t!]
	\centering
	\begin{tabular}{||c || c || c || c ||} 
		\hline
		$d$ & $\mathsf{Del}_{/S_2}(n,d;S_2)$ & $\mathsf{GenSph}(n,d;S_2)$ & $\mathsf{Del}(n,d)$\\ [0.5ex] 
		\hline\hline
		$2$ & $64$ & $64$ & $4096$\\ 
		\hline
		$3$ & $45.255$ & $64$ & $512$\\ 
		\hline
		$4$ & $45.255$ & $64$ & $292.571$\\ 
		\hline
		$5$ & $22.627$ & $64$ & $64$\\ 
		\hline
		$6$ & $17.889$ & $64$ & $40$\\ 
		\hline
		$7$ & $5.657$ & $32$ & $8$\\ 
		\hline
		$8$ & $4.619$ & $32$ & $5.333$\\ 
		\hline
		$9$ & $2.828$ & $16$ & $3.333$\\ 
		\hline
		$10$ & $2.619$ & $16$ & $2.857$\\ 
		\hline
	\end{tabular}
	\caption{Table of values of optimal values of the symmetrized $\mathsf{Del}_{/G_{S_2}}(n,d;S_2)$ LP, the generalized sphere packing bound LP $\mathsf{GenSph}(n,d;S_2)$ in \cite{fazeli} and \cite{cullina}, and the $\mathsf{Del}(n,d)$ LP, for $n=13$ and varying values of $d$.}
	\label{tab:2c}
\end{table}
\subsection{Constant Subblock Composition Constraint}
\label{sec:subblock}
We now move on to studying the constant subblock-composition CSC$_z^p$ constraint, which requires that each one of the $p$ ``subblocks'' of a binary sequence have a constant number, $z$, of $1$s. In particular, for any sequence $\mathbf{x}\in \{0,1\}^n$, we first partition the $n$ coordinates into $p$ subblocks, with the $\ell^\text{th}$ subblock being the vector of symbols $\mathbf{x}_\ell:=\left(x_i\in \{0,1\}:\ \frac{(\ell-1)n}{p}+1\leq i\leq \frac{\ell n}{p}\right)$, for $1\leq \ell\leq p$. We implicitly assume that $p$ divides $n$. Note that hence $\mathbf{x} = \mathbf{x}_1 \mathbf{x}_2\ldots \mathbf{x}_p$. A binary sequence $\mathbf{x}$ respects the CSC$_z^p$ constraint if $w(\mathbf{x}_\ell) = z$, for all $1\leq \ell\leq p$. We let $C_z^{p,(n)}$ (or simply, $C_z^p$) denote the set of all CSC$_z^p$-constrained sequences of length $n$. CSC$_{z}^p$-constrained sequences were introduced in \cite{subblock1} for simultaneous information and energy transfer from a powered transmitter to an energy harversting receiver, while ensuring that the receiver battery does not drain out during periods of low signal energy. The applications of such constrained codes to visible light \cite{subblock2} and powerline communications \cite{subblock3} have also been investigated. {We mention also that closed-form formulae and efficient algorithms for the generalized sphere packing bounds for this constraint were derived in \cite{gspsubblock}}.

As before, we are interested in computing the Fourier coefficients of the function $\mathds{1}_{C_{z}^p}: \{0,1\}^n \to \{0,1\}$. The lemma below provides these Fourier coefficients. 
\begin{lemma}
	\label{lem:fcoeffsubblock}
	For $\mathbf{s}\in \{0,1\}^n$ with $\mathbf{s} = \mathbf{s}_1 \mathbf{s}_2\ldots \mathbf{s}_p$, we have that
	\[
	2^n\cdot\widehat{\mathds{1}_{C_{z}^p}}(\mathbf{s}) = \prod_{\ell = 1}^p K_z^{(n/p)}(w(\mathbf{s}_\ell)),
	\]
	where $K_i^{(n/p)}(j) = \sum_{t=0}^{i} (-1)^t {j\choose t} {n/p-j\choose i-t}$ is the $i^\text{th}$-Krawtchouk polynomial, for the length $n/p$.
\end{lemma}
The proof of the lemma above can be found in Section {III} of the supplement. In Section {IV} of the supplement, we provide detailed analysis of how Lemma IV.3 and Theorem III.1 can be used to reduce the computational complexity of counting the sizes of subblock-constrained subcodes of Reed-Muller codes, for {selected} values of $p$.


Next, we provide a more explicit form of $\mathsf{Del}_{/G_\mathcal{A}}(n,d;\mathcal{A})$, when $\mathcal{A} = C_z^p$, for a fixed blocklength $n$ and parameters $p$ and $z$. {The structure of the symmetry group $G_{C_z^p}$ was derived in \cite{gspsubblock} (see the group $H$ in Section III of \cite{gspsubblock}), but to make the exposition self-contained}, we recall that the symmetry group $G_{C_z^p}$ is generated (via compositions) by the following permutations:
\begin{enumerate}
	\item For $1\leq \ell\leq p$, and $\frac{(\ell-1)n}{p}+1\leq j\leq \frac{\ell n}{p}$, define $\pi_\ell^{\text{perm},j}: [n]\to [n]$ such that $\pi_\ell^{\text{perm},j}$ swaps the indices $\frac{(\ell-1)n}{p}+1$ and $j$, and leaves the other indices in $[n]$ unchanged. 
	\item For $1\leq \ell, \ell^\prime \leq p$, define $\pi_{\ell,\ell^\prime}^{\text{exch}}: [n]\to [n]$ such that $\pi_{\ell,\ell^\prime}^{\text{exch}}$ swaps the element $\frac{(\ell-1)n}{p}+j$ with $\frac{(\ell^\prime-1)n}{p}+j$, for all $1\leq j\leq \frac{n}{p}$, and leaves the other indices in $[n]$ unchanged. In other words, $\pi_{\ell,\ell^\prime}^{\text{exch}}$ exchanges entire blocks indexed by $\ell$ and $\ell^\prime$. 
 \end{enumerate}
 
 {Note that for a fixed block indexed by $1\leq \ell\leq p$, the collection of permutations $\bigl\{\pi_\ell^{\text{perm},j}:$ $\frac{(\ell-1)n}{p}+1\leq j\leq \frac{\ell n}{p} \bigr\}$ generates a group isomorphic to the symmetric group $S_{n/p}$, which contains all permutations of the indices $ \frac{(\ell-1)n}{p}+1\leq i\leq \frac{\ell n}{p}$. Also, the collection of permutations $\{\pi_{\ell,\ell^\prime}^{\text{exch}}: 1 \le \ell,\ell' \le p\}$ generates a group isomorphic to the symmetric group $S_p$.}
	
From the description of the symmetry group $G_{C_z^p}$ above, we arrive at the fact that the orbits of the symmetry group are in one-one correspondence with \emph{unordered} $p$-tuples $\boldsymbol{\alpha}\in \left[0:\frac{n}{p}\right]^p$. Indeed, a given sequence $\mathbf{x}\in \{0,1\}^n$ lies in the orbit $\boldsymbol{\alpha}(\mathbf{x}) = (\alpha_1(\mathbf{x}),\ldots,\alpha_p(\mathbf{x}))$, where wt$(\mathbf{x}_\ell) = \alpha_{\sigma(\ell)}$, for $1\leq \ell\leq p$ and some permutation $\sigma\in S_{p}$. Note hence that the number of orbits, and therefore the sum of the number of variables and the number of constraints in $\mathsf{Del}_{/G_{C_z^p}}(n,d;C_z^p)$ is bounded above by $c\cdot\left(\frac{n}{p}\right)^p$, for some constant $c>0$, which is only a polynomial function of the blocklength. Further, for a given orbit $\boldsymbol{\alpha}$, we let $\mathbf{x}_{\boldsymbol{\alpha}}$ be a representative element of the orbit. In particular, we define $\mathbf{x}_{\boldsymbol{\alpha}}$ to be the concatenation $ \mathbf{x}_{\boldsymbol{\alpha},1}\mathbf{x}_{\boldsymbol{\alpha},2}\ldots\mathbf{x}_{\boldsymbol{\alpha},p}$, with 
\begin{equation}
	\label{eq:orbitrepsubblock}
\mathbf{x}_{\boldsymbol{\alpha},\ell} = (\underbrace{1,1,\ldots,1}_{\alpha_1 \text{ such }},0,0,\ldots,0)
\end{equation}
being of length $n/p$, for $1\leq\ell\leq p$. We thus obtain the following lemma:

\begin{lemma}
	For given orbits $\boldsymbol{\alpha}, \tilde{\boldsymbol{\alpha}}$, with $\mathbf{s}_{\tilde{\boldsymbol{\alpha}}}$ being an orbit representative of $\tilde{\boldsymbol{\alpha}}$, we have
	\[
	2^n\cdot \widehat{\mathds{1}_{\boldsymbol{\alpha}}}(\mathbf{s}_{\tilde{\boldsymbol{\alpha}}}) = \prod_{\ell=1}^p K_{\alpha_\ell}^{(n/p)}(\tilde{\alpha}_\ell),
	\]
where for a given length $m$, $K_i^{(m)}$ is the $i^{\text{th}}$ Krawtchouk polynomial, with $K_i^{(m)}(j)= \sum_{t=0}^{i} (-1)^t {j\choose t} {m-j\choose i-t}$.
\end{lemma}
The proof of the above lemma is similar to the proof of Lemma \ref{lem:fcoeffsubblock} (see Section II of the supplement), and is hence omitted.

Tables \ref{tab:sub2} and \ref{tab:sub3} show numerical evaluations of $\mathsf{Del}_{/G_{C_z^p}}(n,d;C_z^p)$, when $n=15$, and $n=18$, respectively, for fixed parameters $p$ and $z$, and for varying values of $d$. In Table \ref{tab:sub2}, we again compare with upper bounds via the generalized sphere packing bound of \cite{cullina} and \cite{fazeli}. Here too our LP provides tighter upper bounds than the generalized sphere packing bounds.


\begin{table}[t!]
	\centering
	\begin{tabular}{||c || c || c ||} 
		\hline
		$d$ & $\mathsf{Del}_{C_2^3}(n,d;C_2^3)$ & $\mathsf{GenSph}(n,d;C_2^3)$\\ [0.5ex] 
		\hline\hline
		$2$ & $1000$ & $1000$\\ 
		\hline
		$3$ & $826.236$ & $1000$\\ 
		\hline
		$4$ & $826.236$ & $1000$\\ 
		\hline
		$5$ & $157.767$ & $333.333$\\ 
		\hline
		$6$ & $110.851$ & $333.333$\\ 
		\hline
		$7$ & $22.627$ & $166.667$\\ 
		\hline
	\end{tabular}
	\caption{Table of values of optimal values of the $\mathsf{Del}_{/C_2^3}(n,d;C_2^3)$ LP, and the generalized sphere packing bound LP $\mathsf{GenSph}(n,d;C_2^3)$, for $(n,p,z) = (15,3,2)$, and varying values of $d$.}
	\label{tab:sub2}
\end{table}

\begin{table}[t!]
	\centering
	\begin{tabular}{||c || c || c || c ||c ||c ||c ||c ||} 
		\hline
		$d$ & $3$ & $4$ & $5$ & $6$ & $7$ & $8$ & $9$ \\ [0.5ex] 
		\hline\hline
		$\mathsf{Del}_{/C_2^2}(n,d;C_2^2)$ & $556.38$ & $556.38$ & $227.111$ & $165.247$ & $38.118$ & $28.540$ & $4.472$\\ 
		\hline
	\end{tabular}
	\caption{Table of values of optimal values of the $\mathsf{Del}_{/C_2^2}(n,d;C_2^2)$ LP, for $(n,p,z) = (18,2,2)$, and varying values of $d$.}
	\label{tab:sub3}
\end{table}
\subsection{Runlength-Limited (RLL) Constraints}
\label{sec:dinf}
In this subsection, we shall work with runlength-limited constraints on binary sequences. Unlike in the previous subsections, where the Fourier coefficients of the indicator functions of the constraints were explicitly (or analytically) computable, in the application of Theorem \ref{thm:lincount}, for the constraints considered in this section, we shall provide recurrence relations for the Fourier coefficients, which allow them to be efficiently computable, numerically.

We concern ourselves with the  $(d,\infty)$-runlength limited (RLL) constraint. This constraint mandates that there be at least $d$ $0$s between every pair of successive $1$s in the binary input sequence, where $d\geq 1$. 
The $(d,\infty)$-RLL constraint is a special case of the $(d,k)$-RLL constraint, which admits only binary sequences in which successive $1$s are separated by at least $d$ $0$s, and the length of any run of $0$s is at most $k$. Such constraints help alleviate inter-symbol interference (ISI) between voltage responses corresponding to the magnetic transitions, in magnetic recording systems (see \cite{Immink2}). 
We let $S^d$ denote the set of $(d,\infty)$-RLL constrained binary words of length $n$.

Now, for $n\geq 1$, and for $\mathbf{s}\in \{0,1\}^n$, let $\widehat{\mathds{1}_{S^{d}}}^{(n)}(\mathbf{s})$ denote the Fourier coefficient at $\mathbf{s}$, when the blocklength is $n$. We then have that:
\begin{lemma}
\label{lem:fcoeffdinf}
For $n\geq d+2$ and for $\mathbf{s} = (s_1,\ldots,s_n)\in \{0,1\}^n$, we have
\[
\widehat{\mathds{1}_{S^{d}}}^{(n)}(\mathbf{s}) = 2^{-1}\cdot \widehat{\mathds{1}_{S^{d}}}^{(n-1)}\left(s_2^n\right)+(-1)^{s_1}\cdot 2^{-(d+1)}\cdot \widehat{\mathds{1}_{S^{d}}}^{(n-d-1)}\left(s_{d+2}^n\right).
\]
\end{lemma}
\begin{proof}
	We first write
	\begin{align}
		\widehat{\mathds{1}_{S^{d}}}^{(n)}(\mathbf{s}) &= \frac{1}{2^n}\cdot \sum_{\mathbf{x}\in S^d} (-1)^{\mathbf{x}\cdot \mathbf{s}} \notag\\
		&= 2^{-n}\cdot\left(\#\{x^n\in S^d:\ w_\mathbf{s}(x^n)\text{ is even}\} - \#\{x^n\in S^d:\ w_\mathbf{s}(x^n)\text{ is odd}\}\right). \label{eq:inter8}
	\end{align}
	We now prove the recurrence relation when $s_1=0$. Observe that in this case,
	\begin{align}
		&\#\{x^n\in S^d:\ w_\mathbf{s}(x^n)\text{ is even}\} \notag\\ &= \#\{x^n\in S^d:\ w_\mathbf{s}(x^n)\text{ is even and $x_1 = 0$}\}\ +  \#\{x^n\in S^d:\ w_\mathbf{s}(x^n)\text{ is even and $x_1 = 1$}\} \notag\\
		&\stackrel{(a)}{=} \#\{x_2^n\in S^d:\ w_{s_2^n}(x_2^n)\text{ is even}\}\ + 
		\#\{x^n\in S^d:\ w_\mathbf{s}(x^n)\text{ is even and $x_1^{(d+1)} = 10^d$}\} \notag\\
		&=  \#\{x_2^n\in S^d:\ w_{s_2^n}(x_2^n)\text{ is even}\} + \#\{x_{d+2}^n\in S^d:\ w_{s_{d+2}^n}({x_{d+2}^n})\text{ is even}\}, \label{eq:inter9}
	\end{align}
	where (a) holds because $s_1 = 0$ and from the fact that the $(d,\infty)$-RLL constraint requires that $x_2^{d+1} = 0^d$, if $x_1 = 1$. Similarly, we obtain that
	\begin{align}
		\#\{x^n\in S^d:\ w_\mathbf{s}(x^n)\text{ is odd}\}
		&=  \#\{x_2^n\in S^d:\ w_{s_2^n}(x_2^n)\text{ is odd}\} + \#\{x_{d+2}^n\in S^d:\ w_{s_{d+2}^n}({x_{d+2}^n})\text{ is odd}\}. \label{eq:inter10}
	\end{align}
	Now, observe that
	\begin{equation}
		\label{eq:inter14}
		\widehat{\mathds{1}_{S^{d}}}^{(n-1)}(s_2^n) = 2^{-(n-1)}\cdot\left(\#\{x_2^n\in S^d:\ w_{s_2^n}(x_2^n)\text{ is even}\} - \#\{x_2^n\in S^d:\ w_{s_2^n}(x_2^n)\text{ is odd}\}\right)
	\end{equation}
	and that
	\begin{align}
		\widehat{\mathds{1}_{S^{d}}}^{(n-d-1)}(s_{d+2}^n) = 2^{-(n-d-1)}\cdot\Big(\#\{x_{d+2}^n\in S^d:\ w_{s_{d+2}^n}(x_{d+2}^n)&\text{ is even}\} - \notag\\& \#\{x_{d+2}^n\in S^d:\ w_{s_{d+2}^n}(x_{d+2}^n)\text{ is odd}\}\Big).	\label{eq:inter15}
	\end{align}
	Substituting \eqref{eq:inter9} and \eqref{eq:inter10} in \eqref{eq:inter8} and using \eqref{eq:inter14} and \eqref{eq:inter15}, we obtain the recurrence relation when $s_1=0$. The case when $s_1=1$ is proved by similar arguments.
\end{proof}
We shall now explain how Lemma \ref{lem:fcoeffdinf} helps compute the Fourier coefficients for a given (large) $n$, efficiently. First, we note that a direct computation of all the Fourier coefficients of $\mathds{1}_{S^d}$ at blocklength $n$, can be accomplished by the fast Walsh-Hadamard transform (FWHT) algorithm (see Exercise 1.12(b) in \cite{ryanodonnell}), in time $n\cdot 2^n$. Now, let us assume that we pre-compute and store the Fourier coefficients $\left(\widehat{\mathds{1}_{S^{d}}}^{(m)}(\mathbf{s}):\ \mathbf{s}\in \{0,1\}^m\right)$, for $1\leq m\leq d+1$. These Fourier coefficients help initialize the recurrences in Lemma \ref{lem:fcoeffdinf}. Now, given a fixed (large) $n$, the Fourier coefficients at which blocklength we intend computing, we shall calculate, using the recurrence relations above, the Fourier coefficients at all blocklengths $d+2\leq m\leq n$, iteratively, beginning at length $d+2$, and increasing $m$. Assuming that the additions and multiplications in Lemma \ref{lem:fcoeffdinf} take unit time, it can be seen that the time complexity of computing the Fourier coefficient at length $n$ grows as $\sum_{d+2}^n 2^i < 2^{n+1}$. This is much less than the time that is $2^{n+\log_2 n}$, taken by the FWHT algorithm.

However, there still remains the issue of storage cost: at a blocklength $m$, one needs to store all $2^m$ Fourier coefficients in order to facilitate computation of the Fourier coefficients at blocklengths $n>m$. Hence, assuming that the storage of a single Fourier coefficient takes up one unit of space, we see that we require at least $2^n$ units of memory in order to store the Fourier coefficients at blocklength $n$. 

We now use the Fourier coefficients that are numerically computed using Lemma \ref{lem:fcoeffdinf}, to calculate, in Table \ref{tab:rm1inf}, the number of $(1,\infty)$-RLL constrained codewords in select codes, by applying Theorem \ref{thm:lincount}. We denote the binary Hamming code of blocklength $2^t-1$ as $\text{Ham}_t$. 


Next, we obtain upper bounds on the sizes of $(d,\infty)$-RLL constrained codes with a given minimum distance, by directly running the $\mathsf{Del}(n,d;S^d)$ LP. 
Table \ref{tab2} shows comparisons between the upper bounds on $A(n,d;S^2)$, obtained using our  $\mathsf{Del}(n,d;S^2)$ LP, with the generalized sphere packing bound of \cite{cullina} and \cite{fazeli}, when $n=10$, and for varying values of the minimum distance $d$. We also compare these upper bounds with the optimal value of $\mathsf{Del}(n,d)$, since this is a trivial upper bound on $A(n,d;\mathcal{A})$, for any $\mathcal{A}\subseteq \{0,1\}^n$. Note that, from the numerical trials, for certain values of $d$, the generalized sphere-packing bound returns a value that is larger (and hence worse) than the value of $\mathsf{Del}(n,d)$, whereas the optimal value of our $\mathsf{Del}(n,d;\mathcal{A})$ LP is uniformly bounded above by $\mathsf{Del}(n,d)$.

\begin{table}[t!]
	\centering
	\begin{tabular}{||c || c|| c||c||c||} 
		\hline
		$\mathcal{C}$ & RM$(4,2)$ & RM$(4,3)$ &Ham$_3$ &Ham$_4$\\
		\hline
		$N(\mathcal{C};S_1)$& $83$& $1292$& $4$& $101$\\
		\hline
	\end{tabular}
	\caption{Table of values of $N(\mathcal{C};S^1)$, for select codes $\mathcal{C}$}
	\label{tab:rm1inf}
\end{table}

\begin{table}[t!]
	\centering
	\begin{tabular}{||c || c || c || c ||} 
		\hline
		$d$ & $\mathsf{Del}(n,d;S^2)$ & $\mathsf{GenSph}(n,d;S^2)$ & $\mathsf{Del}(n,d)$\\ [0.5ex] 
		\hline\hline
		$2$ & $49.578$ & $60$ & $512$\\ 
		\hline
		$3$ & $32.075$ & $46.5$ & $85.333$\\ 
		\hline
		$4$ & $21.721$ & $46.5$ & $42.667$\\ 
		\hline
		$5$ & $7.856$ & $34$ & $12$\\ 
		\hline
		$6$ & $4.899$ & $34$ & $6$\\ 
		\hline
		$7$ & $2.529$ & $19$ & $3.2$\\ 
		\hline
	\end{tabular}
	\caption{Table of values of optimal values of the $\mathsf{Del}(n,d;S^2)$ LP, the generalized sphere packing bound LP $\mathsf{GenSph}(n,d;S^2)$, and the $\mathsf{Del}(n,d)$ LP, for $n=10$ and varying values of $d$.}
	\label{tab2}
\end{table}

\section{Conclusion}
\label{sec:conclusion}
In this work, we took two approaches to the problem of estimating the sizes of binary error-correcting constrained codes. First, motivated by the application of transmission of codes over stochastic, symmetric, channel noise models---a problem for which explicit capacity-achieving linear codes have been constructed, we consider the question of computing the sizes of constrained subcodes of linear codes. Such constrained subcodes of capacity-achieving linear codes, for example, are resilient to symmetric errors and erasures, in that their error probabilities using the same decoding strategy as for the larger linear code, vanish as the blocklength of the code goes to infinity. Our approach was to view the problem through a Fourier-analytic lens, thereby transforming it into a counting problem in the space of the dual code. As part of our method, we analyzed (analytically or numerically) the Fourier transform of the indicator function of the constraint, and observed the somewhat surprising fact that for many constraints of interest, this Fourier transform is in fact efficiently computable. We then provided values of the number of constrained codewords in select linear codes and algorithmic procedures for efficient counting, in the cases of certain constraints. 

Next, we considered the scenario where the constrained codes were subjected to adversarial bit-flip errors or erasures, with a combinatorial bound on the number of errors or erasures that can be induced. We then proposed numerical upper bounds on the sizes of constrained codes with a given resilience to such combinatorial errors and erasures (equivalently, with a prescribed minimum Hamming distance), via an extension of Delsarte's linear program (LP). 
We observed that the optimal numerical values returned by our LP for different constaints are better than those provided by the generalized sphere packing bounds of Fazeli, Vardy, and Yaakobi (2015).

There are many interesting directions for future work. One line of study would be to build on the Fourier-theoretic techniques in this paper and study the asymptotics (in the limit as the blocklength goes to infinity) of the rates of constrained subcodes of specific linear codes of a given rate $R\in (0,1)$. In parallel, it will be of interest to derive efficient procedures for computing the Fourier transforms of indicator functions of large families of structured constraints. {Similarly, one could try to derive a dual LP formulation and use Fourier-analytic techniques} (see \cite{sam1,sam2}) to derive asymptotic upper bounds on the rate-distance tradeoff for constrained codes. This, for example, will help us understand if the Gilbert-Varshamov lower bounds of Marcus and Roth (1992) are tight for any constrained system. Another direction of work could study the extension of results here to codes with larger alphabet sizes.

\section*{Acknowledgements}
The authors thank Prof. H. D. Pfister and Prof. Manjunath Krishnapur for useful discussions.



\ifCLASSOPTIONcaptionsoff
  \newpage
\fi



%
\bibliographystyle{IEEEtran}
{\footnotesize
	\bibliography{references}}

\end{document}


\title{Supplement to: Estimating the Sizes of Binary Error-Correcting Constrained Codes}
	%
	%
	%
	
	\author{V.~Arvind~Rameshwar,~\IEEEmembership{Student Member,~IEEE,}
		and~Navin~Kashyap,~\IEEEmembership{Senior~Member,~IEEE}
		\thanks{The authors are with the Department of Electrical Communication Engineering, Indian Institute of Science, Bengaluru 560012, India (e-mail: vrameshwar@iisc.ac.in;~nkashyap@iisc.ac.in).}
	}
	
	
	\maketitle
	
	\section{Obtaining MacWilliams' Identities for Linear Codes Via Theorem III.1}
\label{sec:appc}
Consider the simple constraint that admits only sequences having a fixed weight $i\in [0:n]$, where $n$ is the blocklength of the code. Note that in this case, the set of constrained sequences is $\mathcal{A} = W_i$, {where $W_i$ is the set of length-$n$ sequences of weight $0\leq i\leq n$}. By applying Theorem III.1 to this constraint, for a given linear code $\mathcal{C}$, we obtain the well-known MacWilliams' identities \cite{macwilliams} for linear codes. We use the notation $a_i(\mathcal{C})$ for  the number of codewords of weight $i\in [0:n]$ in $\mathcal{C}$, which equals $N(\mathcal{C}; W_i)$, following the notation of Theorem III.1 of the paper.
\begin{theorem}[MacWilliams' identities]
	It is true that
	\[
	a_i(\mathcal{C}) = \frac{1}{\left\lvert \mathcal{C}^\perp \right\rvert}\sum_{j=0}^n K_i^{(n)}(j)\cdot a_j(\mathcal{C}^\perp).
	\]
\end{theorem}
\begin{proof}
	The proof simply uses the fact that $\widehat{\mathds{1}_{W_i}}(\mathbf{s}) = \frac{K_i^{(n)}(w(\mathbf{s}))}{2^n}$. By simplifying the summation in Theorem III.1 of the paper, and by using the fact that $\left\lvert \mathcal{C} \right\rvert\cdot \left\lvert \mathcal{C}^\perp \right\rvert = 2^n$, we obtain the required result.
\end{proof}
	{\section{Proof of Corollary IV.1}}
\label{sec:apphamm}
\begin{proof}[Proof of Corollary IV.1]
	The dual code $\mathcal{C}^\perp$ of the $[2^m-1,2^m-1-m]$ Hamming code is the $[2^m-1,m]$ simplex code, all of whose non-zero codewords (i.e., codewords that are not equal to $0^{2^m-1}$) are of weight $2^{m-1}$. Further, a generator matrix of the simplex code under consideration is $H_{\text{Ham}}$, {where we recall that $H_{\text{Ham}}$ is a parity-check matrix of the Hamming code, defined as follows (see also the end of Section III.A of the paper)}. Let the columns of $H_{\text{Ham}}$ be indexed by $m$-tuples $(x_1,\ldots,x_m)\in \{0,1\}^m \setminus \{0^m\}$, ordered in the standard lexicographic order, i.e., the $i^{\text{th}}$ column of $G$ is indexed as $\mathbf{B}_m(i)$, for $1\leq i\leq 2^m-1$. It is well-known (see, for example, Section 1.10 of \cite{verapless}) that the $j^\text{th}$ row of $H_{\text{Ham}}$ is the evaluation vector, over the $m$-tuples indexing the columns, of the monomial $x_{j}$, for $1\leq j\leq m$. We write this row as Eval$^{\setminus \mathbf{0}}(x_j)$.
	
	Consider the first $m-1$ rows of $H_{\text{Ham}}$, which are the evaluation vectors Eval$^{\setminus \mathbf{0}}(x_j)$, for $1\leq j\leq m-1$. It can be checked that the Hamming weight, $2^{m-1}$, of any of these rows is a multiple of $4$, when $m\geq 3$. Moreover, in any of these rows, if the entry corresponding to the evaluation point $(x_1,\ldots,x_{m-1},0)$ equals $1$, then so does the entry corresponding to the evaluation point $(x_1,\ldots,x_{m-1},1)$. The above two facts imply that each of the first $m-1$ rows of $H_{\text{Ham}}$ can be written as a linear combination of an \emph{even} number of vectors $\mathbf{b}_\ell \in \mathcal{B}$, for $\ell\in [1: \left\lceil \frac{2^m-1}{2}\right \rceil - 1]$. Hence, from Lemma IV.1, it holds that the Fourier coefficient $\widehat{\mathds{1}_{S_{2}}}\left(\text{Eval}^{\setminus \mathbf{0}}(x_j)\right) = 2^{\left\lfloor \frac{2^m-1}{2}\right \rfloor - (2^m-1)}$, for all $1\leq j\leq m-1$. Furthermore, observe that the above arguments also hold for any linear combination of the first $m-1$ rows of $H_{\text{Ham}}$, i.e., it holds that $\widehat{\mathds{1}_{S_{2}}}\left(\text{Eval}^{\setminus \mathbf{0}}(\mathbf{s})\right) = 2^{\left\lfloor \frac{2^m-1}{2}\right \rfloor - (2^m-1)}$, where $\mathbf{s} = \sum_{j=2}^{m}c_j\cdot \text{Eval}^{\setminus \mathbf{0}}(x_j)$, for $c_j\in \{0,1\}$, $j\in [2:m]$.
	
	It can also be seen that since Eval$^{\setminus \mathbf{0}}(x_m)\notin V_{\mathcal{B}}$, we have that $\widehat{\mathds{1}_{S_{2}}}\left(\text{Eval}^{\setminus \mathbf{0}}(x_m)\right) = 0$, and similarly, that $\widehat{\mathds{1}_{S_{2}}}\left(\text{Eval}^{\setminus \mathbf{0}}(\mathbf{s})\right) = 0$, where $\mathbf{s} = \text{Eval}(x_1)+\sum_{j=1}^{m-1}c_j\cdot \text{Eval}(x_j)$, for $c_j\in \{0,1\}$, $j\in [m-1]$. Putting everything together, we observe that for half of the codewords $\mathbf{s}\in \mathcal{C}^\perp$, the Fourier coefficient $\widehat{\mathds{1}_{S_{2}}}(\mathbf{s})$ equals $2^{\left\lfloor \frac{2^m-1}{2}\right \rfloor - (2^m-1)}$, and for another half of the codewords, the Fourier coefficient $\widehat{\mathds{1}_{S_{2}}}(\mathbf{s})$ equals zero. Applying (8) in the paper, we get that
	\begin{align*}
		N(\mathcal{C};S_2) &= |\mathcal{C}|\cdot  \sum_{\mathbf{s}\in \mathcal{C}^\perp} \widehat{\mathds{1}_{S_{2}}}(\mathbf{s})\\
		&= 2^{2^m-1-m}\cdot 2^{m-1}\cdot 2^{\left\lfloor \frac{2^m-1}{2}\right \rfloor - (2^m-1)} = 2^{\left\lfloor \frac{2^m-1}{2}\right \rfloor-1},
	\end{align*}
	where the second inequality holds since $|\mathcal{C}| = 2^{2^m-1-m}$ and $\left\lvert \mathcal{C}^\perp \right \rvert = 2^m$, and half the codewords $\mathbf{s}\in \mathcal{C}^\perp$ have nonzero Fourier coefficient $\widehat{\mathds{1}_{S_{2}}}(\mathbf{s})$.
\end{proof}
	\section{Proof of Lemma IV.3}
\begin{proof}
	We have that
	\begin{align*}
		2^n\cdot\widehat{\mathds{1}_{C_{z}^p}}(\mathbf{s}) &= \sum_{\mathbf{x}\in \{0,1\}^n:\ \mathbf{x}\in C_z^p} (-1)^{\mathbf{x}\cdot \mathbf{s}}\\
		&= \sum_{\mathbf{x}_1\in \{0,1\}^{n/p}:\ w({\mathbf{x}_1}) = z}\ldots \sum_{\mathbf{x}_p\in \{0,1\}^{n/p}:\ w({\mathbf{x}_p}) = z}(-1)^{{\mathbf{x}_1}\cdot \mathbf{s}_{1}}\ldots(-1)^{{\mathbf{x}_p}\cdot \mathbf{s}_{p}}\\
		&= \prod_{\ell=1}^p \left(\sum_{\mathbf{x}_\ell\in \{0,1\}^{n/p}:\ w({\mathbf{x}_\ell}) = z} (-1)^{{\mathbf{x}_\ell}\cdot \mathbf{s}_{\ell}}\right).
	\end{align*}
Now, for any $\ell \in [p]$, the summation above only on the weight $w(\mathbf{s}_\ell)$, i.e., for any permutation of coordinates $\pi: \{0,1\}^{n/p}\to \{0,1\}^{n/p}$, it holds that
\begin{align*}
	\sum_{\mathbf{x}\in \{0,1\}^{n/p}:\ w(\mathbf{x}) = z}	(-1)^{\mathbf{x}\cdot (\pi\cdot \mathbf{s}_\ell)} &= \sum_{\mathbf{x}\in \{0,1\}^n:\ w(\mathbf{x}) = z}  (-1)^{(\pi\cdot \mathbf{x})\cdot (\pi\cdot \mathbf{s}_\ell)}\\
	&= \sum_{\mathbf{x}\in \{0,1\}^n:\ w(\mathbf{x}) = z} (-1)^{\mathbf{x}\cdot \mathbf{s}_\ell}.
\end{align*}
Hence, for any $\ell\in [p]$ and for $\mathbf{s}_\ell$ such that $w(\mathbf{s}_\ell) = j$, it suffices that we calculate the summation above at $\mathbf{s}_\ell=\mathbf{s}^\star = (s_1^\star,\ldots,s_{n/p}^\star)$, with $s_1^\star = \ldots = s_{j}^\star = 1$ and $s_{j+1}^\star=\ldots=s_{n/p}^\star = 0$. By a direct computation, it can be checked that the summation above equals $K_{z}^{(n/p)}(w(\mathbf{s}_\ell))$.
\end{proof}
		\section{Subblock-Constrained Subcodes of RM Codes}
\label{sec:appdn}
In what follows, we shall concern ourselves with the application of Lemma IV.3 and Theorem III.1 to calculating the number of subblock constrained codewords in Reed-Muller (RM) codes RM$(m,r)$, for {selected} values of the number of subblocks $p$.

First, we recall an important property of RM codes, which is sometimes called the Plotkin decomposition (see \cite[Chap. 13]{mws} or the survey \cite{rm_survey}): any length-$2^m$ codeword $\mathbf{c}\in \text{RM}(m,r)$ can be written as the concatenation $\mathbf{c} = (\mathbf{u}\mid \mathbf{u}+\mathbf{v})$, where $\mathbf{u}\in \text{RM}(m-1,r)$ and $\mathbf{v}\in \text{RM}(m-1,r-1)$ and the `$+$' operation in $\mathbf{u}+\mathbf{v}$ is over $\mathbb{F}_2^{2^{m-1}}$. Observe that since RM$(m,t)$, for $1\leq t\leq m$, consists of evaluation vectors of Boolean polynomials of degree at most $t$, it holds that RM$(m-1,r-1)\subset $ RM$(m-1,r)$. In what follows, we ensure that $r\geq 1$ and $m$ is large.

Assume, for simplicity, that $p=2$. We then have that for $0\leq z\leq 2^{m-1}$,
\begin{align}
	N\left(\text{RM}(m,r);C_z^2\right) &= \sum_{\mathbf{c}\in \text{RM}(m,r)} \mathds{1}_{C_z^2}(\mathbf{x}) \notag\\
	&= \sum_{\substack{\mathbf{u}\in \text{RM}(m-1,r),\\ \mathbf{v}\in \text{RM}(m-1,r-1)}} \mathds{1}_{W_z}(\mathbf{u})\cdot \mathds{1}_{W_z}(\mathbf{u}+\mathbf{v}), \label{eq:inter7}
\end{align}
where the second equality uses the Plotkin decomposition and the fact that the set $W_z$ consists of sequences of Hamming weight exactly $z$. Further, let $\mathbf{u}_1,\mathbf{u}_2,\ldots,\mathbf{u}_M$ be an enumeration of coset representatives of distinct cosets of RM$(m-1,r-1)$ in RM$(m-1,r)$, where $M = \frac{|\text{RM}(m-1,r)|}{|\text{RM}(m-1,r-1)|} = 2^{{m-1\choose \leq r} - {m-1\choose \leq r-1}} = 2^{{m-1\choose r}}$. In other words, $\mathbf{u}_i$ is a representative of the coset $\mathbf{u}_i + \text{RM}(m-1,r-1)$, with $\mathbf{u}_i\in \text{RM}(m-1,r)$, for $1\leq i\leq M$, where the cosets $\mathbf{u}_j + \text{RM}(m-1,r-1)$, for different values of $j$, are disjoint. Let $A_{\mathbf{u}}(y)$ be the weight enumerator of the coset $\mathbf{u} + \text{RM}(m-1,r-1)$, at the weight $0\leq y\leq 2^{m-1}$, for $\mathbf{u}\in \text{RM}(m-1,r)$. Then, from \eqref{eq:inter7}, we see that
\begin{align}
	N\left(\text{RM}(m,r);C_z^2\right) &= \sum_{\substack{\mathbf{u}\in \text{RM}(m-1,r),\\ \mathbf{v}\in \text{RM}(m-1,r-1)}} \mathds{1}_{W_z}(\mathbf{u})\cdot \mathds{1}_{W_z}(\mathbf{u}+\mathbf{v}) \notag\\
	&= \sum_{\mathbf{u}\in \text{RM}(m-1,r)}\mathds{1}_{W_z}(\mathbf{u})\cdot \sum_{\mathbf{v}\in \text{RM}(m-1,r-1)} \mathds{1}_{W_z}(\mathbf{u}+\mathbf{v}) \notag\\
	&= \sum_{\mathbf{u}\in \text{RM}(m-1,r)}\mathds{1}_{W_z}(\mathbf{u})\cdot A_{\mathbf{u}}(z) \notag\\
	&\stackrel{(a)}{=} \sum_{i=1}^M \sum_{\mathbf{u}\in \mathbf{u}_i+\text{RM}(m-1,r-1)}\mathds{1}_{W_z}(\mathbf{u})\cdot A_{\mathbf{u}_i}(z) = \sum_{i=1}^M \left(A_{\mathbf{u}_i}(z)\right)^2, \label{eq:cosetsum}
\end{align}
where equality (a) uses the fact that any $\mathbf{u}\in \text{RM}(m-1,r)$ belongs to some coset $\mathbf{u}_i + \text{RM}(m-1,r-1)$. 

While equality \eqref{eq:cosetsum} provides a neat method to count the number of constrained codewords $N\left(\text{RM}(m,r);C_z^2\right)$, provided the coset weight enumerators $A_{\mathbf{u}_i}(z)$, $1\leq i\leq M$, are known, observe that in the summation in \eqref{eq:cosetsum}, we need to perform $M-1 = 2^{{m-1\choose r}}-1$ additions. If $r$ is large, the number of such additions can be fairly high. We show next that with the help of Theorem III.1 and Lemma IV.3, it is possible to reduce the number of computations, when $r$ is large. Before we do so, we recall the fact that for $r\leq m-1$, the dual code of RM$(m,r)$ is the code RM$(m,m-r-1)$. We let $\overline{A}_{\mathbf{u}}(y)$ be the weight enumerator of the coset $\mathbf{u} + \text{RM}(m-1,m-r-2)$, at the weight $0\leq y\leq 2^{m-1}$, for $\mathbf{u}\in \text{RM}(m-1,m-r-1)$. Further, we let $\overline{\mathbf{u}}_1,\overline{\mathbf{u}}_2,\ldots,\overline{\mathbf{u}}_{\overline{M}}$ be an enumeration of coset representatives of distinct cosets of RM$(m-1,m-r-2)$ in RM$(m-1,m-r-1)$, where $\overline{M} = \frac{|\text{RM}(m-1,m-r-1)|}{|\text{RM}(m-1,m-r-2)|} = 2^{{m-1\choose m-r-1}}$.

Now, applying Theorem III.1, we see that
\begin{align}
	N\left(\text{RM}(m,r);C_z^2\right) &= \sum_{\mathbf{s}_1\mathbf{s}_2\in \text{RM}(m,m-r-1)} K_z^{(n/2)}(w(\mathbf{s}_1))\cdot K_z^{(n/2)}(w(\mathbf{s}_2)) \notag\\
	&= \sum_{\substack{\mathbf{u}\in \text{RM}(m-1,m-r-1),\\ \mathbf{v}\in \text{RM}(m-1,m-r-2)}} K_z^{(n/2)}(w(\mathbf{u}))\cdot K_z^{(n/2)}(w(\mathbf{u}+\mathbf{v})) \notag\\
	&= \sum_{\mathbf{u}\in \text{RM}(m-1,m-r-1)} K_z^{(n/2)}(w(\mathbf{u}))\cdot \sum_{\mathbf{v}\in \text{RM}(m-1,m-r-2)} K_z^{(n/2)}(w(\mathbf{u}+\mathbf{v})) \notag\\
	&= \sum_{\mathbf{u}\in \text{RM}(m-1,m-r-1)} K_z^{(n/2)}(w(\mathbf{u}))\cdot \sum_{j=0}^{2^{m-1}}\overline{A}_\mathbf{u}(j)\cdot K_z^{(n/2)}(j) \notag\\
	&= \sum_{i=1}^{\overline{M}} \left(\sum_{j=0}^{2^{m-1}}\overline{A}_{\overline{\mathbf{u}}_i}(j)\cdot K_z^{(n/2)}(j)\right)^2. \label{eq:cosetsum2}
\end{align}
Now, observe that using equality \eqref{eq:cosetsum2}, the number of computations required, in the form of summations, assuming that the coset weight enumerators $\overline{A}_{\overline{\mathbf{u}}_i}(\cdot)$ are known, for all $1\leq i\leq \overline{M}$, is $2^{m-1+\overline{M}}-1 = 2^{m-1+{m-1\choose m-r-1}}-1$. Clearly, since for large $r$ (and large $m$), we have that $m-1+{m-1\choose m-r-1}< {m-1\choose r}$, we note the relative ease of calculating $N\left(\text{RM}(m,r);C_z^2\right)$ via \eqref{eq:cosetsum2}, with the aid of Theorem III.1, as compared to using \eqref{eq:cosetsum}. We remark here that the analysis of the number of codewords in RM$(m,r)$ that lie in $C_z^p$ can be extended to values of $p$ that are powers of $2$, by iteratively applying the Plotkin decomposition. Finally, we note that in order to compute the coset weight enumerators required in \eqref{eq:cosetsum} and \eqref{eq:cosetsum2}, one can use the recursive algorithm provided in \cite{yao}, which applies to RM codes, in addition to polar codes.

	\bibliographystyle{IEEEtran}
	{\footnotesize
		\bibliography{references}}